\documentclass[12pt]{article}
\usepackage{tikz}
\usetikzlibrary{positioning, arrows}
\usepackage{amssymb}
\usepackage{amsfonts}
\usepackage{amsmath}
\usepackage{xcolor}
\usepackage{caption}
\usepackage{graphicx}
\usepackage[affil-it]{authblk}

\numberwithin{equation}{section}

\newcommand{\mb}[1]{{\bss{#1}}}

\newcommand{\bs}[1]{{\boldsymbol{#1}}}

\def\la{\lambda}

\newcommand{\dsl}[1]{{\displaystyle{#1}}}

\newcommand{\eps}{\epsilon}

\usepackage{graphicx}
\newcommand{\HH}{{{\mathcal{H}}}}
\newtheorem{theorem}{Theorem}[section]
\newtheorem{prop}[theorem]{Proposition}

\newtheorem{remark}[theorem]{Remark}
\newenvironment{rem}{\begin{remark} \rm}{\end{remark}}

\newcommand{\wit}[1]{{{\widetilde{#1}}}}
\newcommand{\ou}{{\overline{u}}}

\newcommand{\BB}{\mathcal{B}}

\newcommand{\bm}[1]{\mbox{\boldmath{$#1$}}}
\newcommand{\bss}[1]{\boldsymbol{#1}}
\newcommand{\bgamma}{{{\bar{\gamma}}}}
\newcommand{\BcT}{{{\mb{\mathcal{T}}}}}

\renewcommand{\sfdefault}{bch}
\def\barr{\hbox{{\fontfamily{\sfdefault}\selectfont I\hskip -.35ex R}}}
\def\sbarr{\hbox{{\fontfamily{\sfdefault}\selectfont {\scriptsize I}\hskip -.25ex {\scriptsize R}}}}

\newcommand{\RR}{{\sbarr}}

\newcommand{\Sigb}{{{\bs{\Sigma}}}}
\newcommand{\CS}{{\mathcal{S}}}
\newcommand{\CM}{{\mathcal{M}}}
\newcommand{\bna}{{{\bs{\nabla}}}}

\begin{document}
\title{ \medskip Two-layer sharply stratified Euler fluids\\ in three  dimensions:
a Hamiltonian setting}
\author{The authors}
\author{ 
R. Camassa${}^1$, G. Falqui${}^{2,5}$, G. Ortenzi${}^{3,6}$, M. Pedroni${}^{4,5}$, E. Sforza${}^{2,5,7}$\footnote{Corresponding Author: e.sforza4@campus.unimib.it}
\medskip\\
{\small $^1$University of North Carolina, Carolina Center for Interdisciplinary Applied Mathematics,}
\\ 
{\small Department of Mathematics, Chapel Hill, NC 27599, USA}
\medskip\\
{\small  $^2$Department of Mathematics and Applications, University of  Milano-Bicocca,}
\\
{\small Via Roberto Cozzi 55, I-20125 Milano, Italy} 
\medskip\\
{\small  $^3$Dipartimento di Matematica ``G.Peano", Università di Torino,\\ Via Carlo Alberto 10, I-10123 Torino, Italy}
\\ \medskip
{\small $^4$Dipartimento di Ingegneria Gestionale, dell'Informazione e della Produzione,}
\\  
{\small Universit\`a di Bergamo, Viale Marconi 5, I-24044 Dalmine (BG), Italy}
\medskip\\
{\small  $^5$INFN, Sezione di Milano-Bicocca, Piazza della Scienza 3, I-20126 Milano, Italy}\\ \medskip
{\small  $^6$INFN, Sezione di Torino, Via Pietro Giuria 1, I-10125 Torino, Italy}\\ \medskip
{\small  $^7$Joint Ph.D.\ program University of Milano-Bicocca, University of Pavia and INdAM}
}

\medskip

\maketitle

\abstract{\noindent 
Three-dimensional two-layer incompressible Euler fluids are studied from a Hamiltonian perspective.
A natural Hamiltonian structure for the effective 2D model described by the interface-value of the field variables is obtained by means of a Hamiltonian reduction process from the 3D Poisson structure. The problem  of expressing the fluid's energy in terms of the reduced variables is considered, and it is shown that in the weakly non linear approximation our procedure gives rise to a so-called 2D Kaup-Broer-Kupershmidt Boussinesq  (KBK-B) model with {\em ``critical"} parameters.
A model  weakly dependent on one of the two horizontal directions is also discussed, whose unidirectionalization turns out to be the well-known Kadomtsev-Petviashvili (KP) equation. A Dirac-type reduction process of the Hamiltonian structure of the KBK-B model yields a natural Hamiltonian structure for KP {\em qua} 2+1-dimensional model.}

\section{Introduction}
\label{intro}
Density stratification is a crucial component of the dynamics of a myriad fluid phenomena over a wide range of scales, from those pertaining to geophysical applications in the ocean and atmosphere (see e.g.~\cite{Pedlov,Bleck02}), to much smaller ones such as those that can occur in industrial processes~\cite{Homsy, Lucas}. Under gravity, the displacement of fluid parcels from their neutral
buoyancy position often results in internal wave motion, whose mathematical modeling is an active subject of investigation, as the fundamental governing equations, even under the simplifying assumptions of ideal, incompressible fluids, are not directly solvable, and are complex enough to ``hide" some of the basic mechanisms underlying the dynamics. 

In this paper, we focus on three-dimensional settings, whereby constant density surfaces, or  isopycnals,  are viewed as graphs over two horizontal directions orthogonal to gravity, see figure~\ref{fig1}. 
In the literature, these settings have received relatively less attention than their two-dimensional counterparts~(see e.g., \cite{Ben86,CGK05,ChCa99,CMMRT09,CFOPT23}), especially from an analytical and geometrical viewpoint. Our aim is to study the fundamental mathematical structures of the governing equations, and to use these as the starting point for the derivation of simplified mathematical models in a physically relevant setup. Specifically, the fluid is viewed as inviscid and incompressible,  bound by two horizontal plates extending to infinity. Under these conditions, the equations of motion, sometimes referred to as the Euler-Boussinesq system, admit a natural Hamiltonian structure, as presented in~\cite{Bow87}. This formulation, unlike other possibilities in the literature (e.g.~\cite{ZK97,Z85,S88}), avoids the use of potentials and utilizes dependent variables constructed directly from the measurable quantities of density and velocity.  The Hamiltonian structure is arguably the preferred avenue  for investigating  conservation laws of a dynamical system, and is an excellent starting point for a systematic approach to its analysis from the perspective of Hamiltonian geometry that we shall henceforth adopt.
Although other variational formulations of the problem are available in the literature (see, e.g., \cite{CGK05,PCH,BaGaTe07})
this framework is the preferred one to  be exploited for systematic reductions to models that inherit the parent Hamiltonian structure and its conservation laws, and that can shed light on the dynamical evolution by providing 
closed-form computable solutions (e.g. traveling waves).

The following notation will be used throughout the paper:
$$
\bm{x}=
\begin{pmatrix}
x\\y\\z
\end{pmatrix}\,,
\quad
\boldsymbol{\nabla}=
\begin{pmatrix}
\partial_x \\
\partial_y\\
\partial_z
\end{pmatrix}
\equiv
\begin{pmatrix}
\nabla \\
\partial_z
\end{pmatrix}
$$
so that $\nabla$ will refer to the gradient with respect to the two-dimensional horizontal components 
$(x,y)$. The horizontal divergence will be denoted by $\nabla\cdot$, 
and the Laplacian will be understood to be the ordinary definition $\Delta=\partial_x^2+\partial_y^2+\partial_z^2$, which, when applied to functions depending on the horizontal coordinates will collapse to $\partial_x^2+\partial_y^2=\nabla\cdot \nabla=\Delta$.

Three dimensional vector variables will be represented by boldface capital letters, and their components, in lower case, by three different letters, e.g., 
$$
\bm{U}(x,y,z,t)=\begin{pmatrix}
u\\v\\w
\end{pmatrix}\,,
\qquad
\bm{\Sigma}(x,y,z,t)=\begin{pmatrix}
\sigma\\\tau\\\chi
\end{pmatrix}\,,
$$
will be the velocity field and the weighted vorticity, respectively, while the lower index $i=1,2$ will refer to the layer component, with $i=1$ for the upper layer and $i=2$ for the lower one. Values of the dependent variables restricted to the interface $z=\zeta(x,y,t)$ will be adorned by tildes, e.g., 
$$
\wit{\bm{U}}(x,y,t)=\bm{U}\big(x,y,\zeta(x,y,t),t\big)=
\begin{pmatrix}
\wit{u}\\\wit{v}\\\wit{w} 
\end{pmatrix}
\,.
$$

This paper is organized as follows: Section ~\ref{BB} presents a quick overview of the Hamiltonian structure of the 3D Euler equations with density stratification, before moving on to consider the limiting configuration of two (homogeneous) density layers and discuss its Hamiltonian reduction in Section \ref{sectHamred}. Next, Section~\ref{secWNL} focuses on the asymptotic expansion in the long wave limit of this Hamiltonian structure, and extends the Hamiltonian reduction,  ultimately to derive the so-called two-dimensional Kaup-Broer-Kupershmidt Boussinesq  (KBK-B) model governing the weakly nonlinear  (WNL) regimes of dispersive waves propagation at the interface between the two fluids. Special solution classes of this model are discussed in Section~\ref{specsol}, as well as its well-posedness features determined by the dispersion relation around simple equilibria solutions. Section~\ref{KP} continues the process of systematic Hamiltonian reduction in the presence of constraints to impose the nearly unidirectional asymptotic limit of propagation to the KBK-B model. This results in the appropriate version of the celebrated Kadomtsev-Petviashvili (KP) equation for our setup. Conclusions and perspectives on future work are presented in Section \ref{CP}. The relation between the Hamiltonian structures of the parent equations already present in the literature, details on Hamiltonian reduction process, and the connection with the Dirichlet to Neumann approach are summarized in appendices.

\section{The Benjamin-Bowman structure and its reduction}
\label{BB}
We consider a three dimensional stratified Euler fluid contained in an unbounded box 
\begin{equation}
\mathcal{B}=\barr^2\times(-h_2,h_1)\ni (x,y,z)\, ; 
\end{equation}
such a fluid is governed by the incompressible Euler equations for the velocity field $\bss{U}=(u,v,w)$ and non-constant density $\rho=\rho(x,y,z,t)$ in the presence of gravity 
$-g\bss{k}$: 
\begin{equation}
\label{EEq}
 {\rho_t}+\bss{U}\cdot\bm{\nabla}\rho=0, \qquad \bm{\nabla} \cdot \bss{U} =0, \qquad \bss{U}_t+(\bss{U}\cdot\bm{\nabla})\, \bss{U} + \frac{\bm{\nabla} p}{\rho} + g \bss{k}=\bss{0}, 
 \end{equation}
with boundary conditions 
 \begin{equation}\begin{split}
 &\bss{U}(r\to\infty,z,t)=\bss{0},\quad \rho(r\to\infty,z,t)=\rho_0(z), \\ &w(x,y,-h_2,t)=w(x,y,h_1,t)={0,}\,\\ & r=\sqrt{x^2+y^2}\,,
 \quad z\in(-h_2,h_1), \quad t\in \barr^+\,.
 \end{split}
 \label{bEEq}
\end{equation}
We remark that $z=-h_2$ and $z=h_1$ are the locations of the bottom and top confining plates, respectively, and that $\rho_0(z)$ is a reference density, which can at most depend on the vertical coordinate $z$.

Generalizing to the physical three dimensions the results in two dimensions obtained by 
Benjamin in  \cite{Ben86}, Bowman \cite{Bow87} was able to endow equations \eqref{EEq} with a non-standard Hamiltonian structure, which can be described as follows. 
The Hamiltonian variables are to be taken as the pair $(\rho, \bs{\Sigma})$, where $\bs{\Sigma}$ is the {\em weighted} vorticity vector
\begin{equation}\label{Sigmadef}
\bs{\Sigma}=\bm{\nabla}\times (\rho\,\bs{U})\, .
\end{equation}
The Hamiltonian functional is  the total energy
\begin{equation}
\label{Hbow}
\mathcal{H}=\int_\BB \left(\frac12\rho |\mb{U}|^2+g\, z(\rho-\rho_0) \right)d^3{x}\, , 
\end{equation}
and the Hamiltonian (Poisson) operator is  the operator-valued $4\times4$ matrix in the evolution equations 
\begin{equation}
\label{PBow}
\left(
\begin{array}{c}\bigskip
\rho_t\\
\Sigb_t
\end{array}
\right)=
-\left(
\begin{array}{cc}\bigskip
0&\bm{\nabla}\cdot(\rho\, \bm{\nabla}\times\>\> )\\
\bm{\nabla}\times(\rho\, \bm{\nabla}\cdot \>\>)
&\bm{\nabla}\times(\mb{\Sigma}\times \bm{\nabla}\times \> \>)
\end{array}\right)\, \left(
\begin{array}{c}\medskip
\dsl{\frac{\delta \HH}{\delta\rho}}\\
\dsl{\frac{\delta \HH}{\delta \Sigb}}
\end{array}
\right)\, , 
\end{equation}
a system that can be shown~\cite{Bow87} to be  equivalent to the incompressible variable density Euler equations \eqref{EEq}.
\begin{rem} 
A Hamiltonian formulation of the full 3D Euler equations for vanishing helicity flows was found in \cite{Z85, MS85, ZK97} in terms of Clebsch potentials. The Benjamin-Bowman formulation we started from is tailored to the incompressible regime, but has the advantage of being formulated by means of physically measurable quantities, and not restricted to zero helicity. The two formulations are connected by  a suitable ``coordinate" transformation. We further  discuss such a relation in Appendix \ref{App A}.
\end{rem} 
In the spirit of our previous papers~\cite{CFO17,CFOPT23, FS25} we now want to specialize the Hamiltonian structure~\eqref{PBow} to the case of a sharply stratified two-layered configuration, in which a homogeneous fluid of constant density $\rho_2$ lies under another fluid of constant density $\rho_1$ (with 
$\rho_2>\rho_1$ for stability). The two fluids are separated  by an  interface  $z=\zeta(x,y,t)$, i.e., $\zeta$ is a function of the horizontal variables $(x,y)$, and evolves with time $t$, a setup  
which we pictorially represent in Figure~\ref{fig1}.

\begin{figure}[t!]
\centering
\includegraphics[height=5cm]{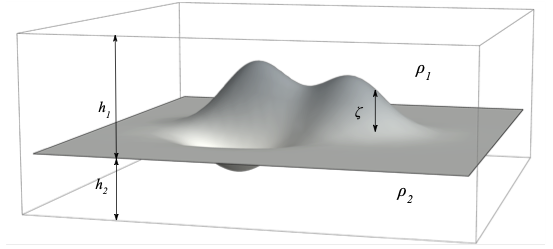}
\captionsetup{width=00.9\textwidth, font=small}
\caption{The geometry of a 3D two-layer configuration}
\label{fig1}
\end{figure}

The Euler variables are expressed via the Heaviside $\theta$-function as
\begin{equation}\label{no1}
\begin{split}
&\rho(x,z,t)=\rho_2+(\rho_1-\rho_2)\theta(z-\zeta(x,y,t))\\
&\mb{U}(x,y,z,t)=\mb{U}_2(x,y,z,t)+(\mb{U}_1(x,y,z,t)-\mb{U}_2(x,y,z,t))\, \theta(z-\zeta(x,y,t)),\\
\end{split}
\end{equation}
where $\mb{U}_1=(u_1,v_1,w_1)$ and $ \mb{U}_2=(u_2,v_2,w_2)$ denote the velocity vector fields in the two domains, while the reference far-field density
is
\begin{equation}\
\label{rho03D}
\rho_0(z)=\rho_2+(\rho_1-\rho_2)\theta(z),\quad z\in(-h_2,h_1)\, .
\end{equation}
Assuming the flow to be potential in the bulk of both domains, i.e., 
\begin{equation}\label{pot}
\mb{U}_i=\bm{\nabla}\bs{\Phi}_i\, ,\quad i=1,2,
\end{equation}
the characteristic variable { $\Sigb$} is a ``vortex sheet'' supported on the interface  
given by
\begin{equation}
\label{Sigmasheet}
\Sigb(x,y,z;t)=\delta(z-\zeta)\, 
\left(
\begin{array}{c}
\rho_2 v_2-\rho_1v_1+\zeta_y(\rho_2 w_2-\rho_1w_1)\\
-\left(\rho_2 u_2-\rho_1u_1+\zeta_x(\rho_2 w_2-\rho_1w_1)\right)\\
\zeta_x(\rho_2 v_2-\rho_1v_1)-\zeta_y(\rho_2 u_2-\rho_1u_1)
\end{array}
\right)
=\delta(z-\zeta)\, \left(
\begin{array}{c}
\wit{\sigma}(x,y,t)\\  \wit{\tau}(x,y,t)\\ \wit{\chi}(x,y,t)
\end{array}
\right),
\end{equation}
where we used the Dirac $\delta$-function distribution, and we have defined the sheet variables according to the notation introduced in~\S\ref{intro}. 
Equation \eqref{Sigmasheet} leads to the following
\begin{prop}\label{propsigma}
The interface vorticity $\wit{\bm\Sigma}=(\wit \sigma, \wit \tau, \wit \chi)$ satisfies the following properties
\begin{equation}
\wit{\chi}=\zeta_x\wit{\sigma}+\zeta_y\wit{\tau}\,,
\qquad
\nabla\cdot
\begin{pmatrix}
\wit{\sigma}\\
\wit{\tau}
\end{pmatrix}
=
\wit{\sigma}_{x}+\wit{\tau}_{y}=0 \,.
\label{propsig}
\end{equation}
\end{prop}
{\bf Proof}. The first property follows directly from \eqref{Sigmasheet}. 

For the second, we recall that by its very definition \eqref{Sigmadef}, $\bm{\nabla}\cdot\Sigb=0$. 
This yields
\begin{equation}\label{eq1}
\begin{split}
&\left[\partial_x(\rho_2 v_2-\rho_1v_1+\zeta_y(\rho_2 w_2-\rho_1w_1))-\partial_y(\rho_2 u_2-\rho_1u_1+\zeta_x(\rho_2 w_2-\rho_1w_1))\right.\\&\left.
+\partial_z(\zeta_x(\rho_2 v_2-\rho_1v_1)-\zeta_y(\rho_2 u_2-\rho_1u_1))\right]\delta(z-\zeta)+\\ &
\left[
-\zeta_x(\rho_2 v_2-\rho_1v_1+\zeta_y(\rho_2 w_2-\rho_1w_1))+\zeta_y (\rho_2 u_2-\rho_1u_1+\zeta_x(\rho_2 w_2-\rho_1w_1))
 \right.
\\& \left. +\zeta_x(\rho_2 v_2-\rho_1v_1)-\zeta_y(\rho_2 u_2-\rho_1u_1)
\right] 
\delta^\prime(z-\zeta) =0\,.
\end{split}
\end{equation}
The vanishing of the $\delta^\prime(z-\zeta)$ coefficient is just a restatement of the first property in~(\ref{propsig}).
By straightforward algebraic manipulations, and taking into account that $\partial_z\zeta=0$,   the $\delta(z-\zeta)$ coefficient of \eqref{eq1} turns into
\begin{equation}\label{e1.5}
\begin{split}
&\left[\partial_x(\rho_2 v_2-\rho_1v_1+\zeta_y(\rho_2 w_2-\rho_1w_1))+\zeta_x\partial_z(\rho_2 v_2-\rho_1v_1)+\right.\\
&\left.
-\partial_y(\rho_2 u_2-\rho_1u_1+\zeta_x(\rho_2 w_2-\rho_1w_1))+\zeta_y\partial_z((\rho_2 u_2-\rho_1u_1))
\right]
\end{split}
\end{equation}
Adding and subtracting the term  $\zeta_x\zeta_y(\rho_2 w_2-\rho_1w_1)$ and judiciously collecting terms we finally get 
\begin{equation}\label{relazdiv}
\wit{\sigma}_{x}+\wit{\tau}_{y}=0\, ,
\end{equation}
that is, the second property in~(\ref{propsig}). \hfill$\square$

In summary, the assumption of a sharp 2-layer stratification and of the absence of bulk vorticity leaves us with a weighted vorticity sheet  
$\wit{\bm\Sigma}$ satisfying 
the ``constraints" of Proposition \ref{propsigma}, i.e., with effective 2D variables $(\zeta, \wit{\sigma},\wit{\tau})$. In the next section we shall reduce the Benjamin-Bowman Hamiltonian operator given by \eqref{PBow} to the ``manifold"  parametrized by these effective variables.
\begin{rem}
It should be remembered that $(\wit{{\sigma}}, \wit{\tau})$  is evaluated at points on the interface. Hence the equation $\wit{\sigma}_x+\wit{\tau}_y=0$ is to be carefully  interpreted. As usual, this equation locally implies the existence of an interface-stream function $\wit{\psi}(x,y,t)$ such that 
\begin{equation}\label{2dstreamfunction}
\wit{\sigma}=\wit{\psi}_y \,, \qquad 
\wit{\tau}=-\wit{\psi}_x\, .
\end{equation}
In what follows, from Section \ref{specsol} onwards, we will make extensive use the fact that the orthogonal $2$-vector $(-\wit{\tau},\wit{\sigma})$ is a potential field.
\end{rem}

\subsection{The Hamiltonian reduction}\label{sectHamred}
As we already discussed, the incompressible non-homogeneous Euler equations for fluids in the box $\BB$ are defined on the manifold (dropping $t$-dependence of the variables for ease of notation)
\begin{equation}\label{M-mfld}
\CM=\left\{(\rho(x,y,z), \Sigb(x,y,z))\right\},\quad (x,y)\in \barr^2, z\in(-h_2,h_1)\, .
\end{equation}
This manifold is endowed with the Poisson tensor
\begin{equation}\label{PBB}
{\cal P}_B=-\left(
\begin{array}{cc}\bigskip
0&\bm{\nabla}\cdot(\rho\, \bm{\nabla}\times\>\> )\\
\bm{\nabla}\times(\rho\, \bm{\nabla}\cdot \>\>)
&\bm{\nabla}\times(\mb{\Sigma}\times \bm{\nabla}\times \> \>)
\end{array}\right)\, .
\end{equation}
We consider the submanifold of two-layer configurations
\begin{equation}\label{S/mfld}
\begin{split}
\CS=&\left\{\rho=\rho_2+(\rho_1-\rho_2)\theta(z-\zeta(x,y)), \quad
\sigma =\wit{\sigma} (x,y)\delta(z-\zeta(x,y)), \right.
\\ & \left.
\quad \tau =\wit{\tau} (x,y)\delta(z-\zeta(x,y)), 
\quad \chi =\wit{\chi} (x,y)\delta(z-\zeta(x,y)),
 \quad \wit{\chi}=\zeta_x\wit{\sigma}+\zeta_y\wit{\tau}\, 
\right\}.
\end{split} 
\end{equation}
In order to equip $\mathcal{S}$ with a natural Hamiltonian structure, we use one of the specializations of the Marsden-Ratiu (MR) Hamiltonian reduction theorem~\cite{MR86}, considering the reduction with respect to the distribution $\mathcal{D}$ defined as
\begin{equation}\label{DD}
\mathcal{D}= 
{\cal P}_B\left((T\CS)^0\right)
\,, 
\end{equation}
where $(T\CS)^0$ is the annihilator of the tangent space $T\CS\subset T\CM$ (this approach parallels what has successfully been applied in the $2D$ case in \cite{CFOPT23, FS25}).
Then we can obtain a Poisson tensor on the quotient $\CS/(P_B(T\CS)^0)$ by means of the projection
\begin{equation}\label{projmap}
\zeta=\dsl{\frac1{\rho_2-\rho_1}}\int_{-h_2}^{h_1}(\rho-\rho_1)\, dz -h_2,\,\quad \wit{\sigma}=\int_{-h_2}^{h_1}{\sigma}\, dz, \, \quad \wit{\tau}=\int_{-h_2}^{h_1}\tau \, dz\, .
\end{equation}
We refer to Appendix \ref{App B} for details. Note the empty intersection 
\begin{equation}
\label{S=N}
{\cal P}_B\left((T\CS)^0\right)\cap T\CS=\{0\}\, , 
\end{equation}
so that the MR reduced Poisson manifold is indeed identical to $\CS$, and thus parametrized by the triple $({\zeta, \wit{\sigma},\wit{\tau}})$. 
In these coordinates the reduced Poisson tensor is given by the matrix of differential operators
\begin{equation}
\label{Pred}
{\cal P}=\left(
\begin{array}{ccc}
 0 & - \partial_y  &\partial_x   \\
 - \partial_y  &0   &   0\\
 \partial_  x&  0 &   0
\end{array}
\right).
\end{equation}
In summary, the geometric Poisson reduction procedure endows the phase space of $2$-layered configurations with a simple Poisson tensor, independent of the density of the two layers as well as other hardware parameters of the model, {leading to the Hamilton equations of motion 
\begin{equation}\label{eqn_WNL0}
    \begin{cases}\medskip 
        &\zeta_t = -\partial_y \dfrac{\delta\mathcal{H}}{\delta\wit{\sigma}}+\partial_x\dfrac{\delta\mathcal{H}}{\delta\wit{\tau}}\\ \medskip
          &\wit{\sigma}_{t} =-\partial_y \dfrac{\delta\mathcal{H}}{\delta\zeta}\\
        &\wit{\tau}_{t} = \partial_x\dfrac{\delta\mathcal{H}}{\delta\zeta}\ , 
    \end{cases}
\end{equation}  
supplemented by the condition $ \wit{\sigma}_{x}+\wit{\tau}_{y}=0$,  where $\mathcal{H}$ is the energy~\eqref{Hbow}.

}
The next step is to express the Hamiltonian of such configurations in terms of the Hamiltonian variables  $({\zeta, \wit{\sigma},\wit{\tau}})$, a task we face in Section~\ref{secWNL}. We can anticipate that in the 3-dimensional setting, this process is quite more involved than in its $2$-dimensional counterpart. However, restricting to the lowest non-trivial order in a natural asymptotic expansion yields  a {\em local\/} Hamiltonian density, leading to the motion equations  
\eqref{eqn_WNL}.

\section{The Weakly-Non Linear  asymptotics and the KBK-Boussinesq  2D model}
\label{secWNL}

In this section we shall first tackle the problem of computing the Hamiltonian of the reduced $2$-layer system whose (simple) Poisson structure was determined in Section \ref{sectHamred}, in terms of the free variables $({\zeta, \wit{\sigma},\wit{\tau}})$. Next, we shall restrict ourselves to the simplest dispersive asymptotic model, which we will denote as the Weakly Non-Linear (WNL) system.

\subsection{Asymptotic expansion of the energy}\label{Asy_ener}
The assumption of  bulk irrotationality of the fluid flow allows us to introduce the bulk velocity potentials $\Phi_i(x,y,z)$. The Taylor series expansions with respect to the vertical variable $z$  of these potentials at the plate locations
$z_i\equiv (-1)^{i-1} h_i$ are
\begin{equation}\label{taylor_Phi}
    \Phi_i(x,y,z) = \sum_{n=0}^{\infty} \frac{(z+(-1)^{i}h_i)^n}{n!}{\partial^n \Phi_i \over \partial z^n}(x,y,z_i)\, .
\end{equation}
Every odd term of these series is zero,
due to incompressibility, i.e., $\Delta\Phi_i = 0 $, and to the kinematic boundary conditions at the bottom and at the top of the channel $\partial_z \Phi_i(x,y,z_i)=0$ (see, e.g., \cite{Wu00}). 
Thus, with $n =2m$ 
\begin{equation}
    \partial_z^{2m} \Phi_{i}(x,y,z_i) = (-1)^m (\partial_x^{2}+\partial_y^{2})^m\Phi_{i \, 0} \,   ,
\end{equation}
where we denoted by $\Phi_{i\,0}$ the bulk potentials  at the boundary, i.e., $\Phi_{i\, 0}=\Phi_i(x,y,z_i)$, 
and~(\ref{taylor_Phi}) reduces to 
\begin{equation}
     \Phi_i(x,y,z) = \sum_{n=0}^{\infty}  \frac{(-1)^n}{(2n)!}(z+(-1)^{i}h_i)^{2n}(\partial_x^{2}+\partial_y^{2})^n\Phi_{i \, 0}\,  .
\end{equation}
The corresponding Taylor expansions of the velocity components are 
\begin{equation}\label{taylor_u}
    \begin{split}
        u_i(x,y,z) &= \sum_{n=0}^{\infty}  \frac{(-1)^n}{(2n)!}(z+(-1)^{i}h_i)^{2n}(\partial_x^{2}+\partial_y^{2})^n u_{i \, 0}\, , \\
        v_i(x,y,z) &= \sum_{n=0}^{\infty}  \frac{(-1)^n}{(2n)!}(z+(-1)^{i}h_i)^{2n}(\partial_x^{2}+\partial_y^{2})^n v_{i \, 0}\, , \\
      { w_i(x,y,z)} &= \sum_{n=0}^{\infty} \frac{(-1)^{n+1}}{(2n+1)!}(z+(-1)^{i}h_i)^{2n+1}(\partial_x^{2}+\partial_y^{2})^n (\partial_x u_{i \, 0}+\partial_y v_{i \, 0})\, ,
    \end{split}
\end{equation}
and their counterpart for the interface velocities $\wit{\textbf{U}}(x,y) = \textbf{U}(x,y,\zeta(x,y))$ are
\begin{equation} \label{taylor_utilde}
    \begin{split}
        \wit{u}_i(x,y) &= \sum_{n=0}^{\infty}  \frac{(-1)^n}{(2n)!}\eta_i^{2n}(\partial_x^{2}+\partial_y^{2})^n u_{i \, 0}\, , \\
        \wit{v}_i(x,y) &= \sum_{n=0}^{\infty}  \frac{(-1)^n}{(2n)!}\eta_i^{2n}(\partial_x^{2}+\partial_y^{2})^n v_{i \, 0}\, , \\
        \wit{w}_i(x,y) &=(-1)^i \sum_{n=0}^{\infty} \frac{(-1)^{n+1}}{(2n+1)!}\eta_i^{2n+1}(\partial_x^{2}+\partial_y^{2})^n (\partial_x u_{i \, 0}+\partial_y v_{i \, 0})\, , 
    \end{split}
\end{equation}
since
\begin{equation}
    \eta_1 = h_1-\zeta, \quad \eta_2=h_2+\zeta\,.
\end{equation}
The power series \eqref{taylor_u} and \eqref{taylor_utilde} can be used as the starting point of  long-wave asymptotics, by introducing rescaled nondimensional independent variables, 
  \begin{equation}
      x = Lx^*, \quad y =Ly^*, \quad z = {h}z^* \,.
  \end{equation}
Here  $L$ is assumed to be the typical horizontal wavelength of the dynamics, 
 and $h=h_1+h_2$ offers a natural vertical scale. 
Long wave asymptotics is defined assuming that the parameter $\epsilon ={{h}}/{L}$ is small. This can also be viewed as the weight of dispersion corrections to the well known shallow water wave theory. Moreover, the maximum displacement $a$ from equilibrium can be used to nondimensionalize  the smooth interface 
$\zeta(x,y)$,  
 \begin{equation}
      \zeta = a \,\zeta^*\,.
  \end{equation}
This defines the parameter $\alpha = {a}/{{h}}$ whose smallness ultimately characterizes the Weakly Nonlinear (WNL) asymptotics to be discussed next.
  
Unless otherwise explicitly stated, from now on we will use non-dimensional variables and drop asterisks, but, when needed,   
with a slight abuse of notation the order symbol $O(\cdot)$ will also denote magnitude of bounded {\it dimensional} quantities whenever this can be done without generating confusion. 
 With the above scalings, the formal power series \eqref{taylor_u} and \eqref{taylor_utilde} become
\begin{equation}\label{taylor_u_eps}
    \begin{split}
        u_i(x,y,z) &= \sum_{n=0}^{\infty}  \frac{(-1)^n}{(2n)!}\eps^{2n}\left(z+(-1)^{i}\frac{h_i}{{h}}\right)^{2n}(\partial_x^{2}+\partial_y^{2})^n u_{i \, 0}\, , \\
        v_i(x,y,z) &= \sum_{n=0}^{\infty}  \frac{(-1)^n}{(2n)!}\eps^{2n}\left(z+(-1)^{i}\frac{h_i}{{h}}\right)^{2n}(\partial_x^{2}+\partial_y^{2})^n v_{i \, 0}\, , \\
        w_i(x,y,z) &= \eps\sum_{n=0}^{\infty} \frac{(-1)^{n+1}}{(2n+1)!}\eps^{2n}\left(z+(-1)^{i}\frac{h_i}{{h}}\right)^{2n+1}(\partial_x^{2}+\partial_y^{2})^n (\partial_x u_{i \, 0}+\partial_y v_{i \, 0})\, , 
    \end{split}
\end{equation}
and \begin{equation}\label{taylor_exp_tilde}
\begin{split}
       \wit{u}_i(x,y) &= \sum_{n=0}^{\infty}  \frac{(-1)^n}{(2n)!}\eps^{2n}\eta_i^{2n}(\partial_x^{2}+\partial_y^{2})^n u_{i \, 0}\, , \\
        \wit{v}_i(x,y) &= \sum_{n=0}^{\infty}  \frac{(-1)^n}{(2n)!}\eps^{2n}\eta_i^{2n}(\partial_x^{2}+\partial_y^{2})^n v_{i \, 0}\, , \\
        \wit{w}_i(x,y) &=(-1)^i \eps \sum_{n=0}^{\infty} \frac{(-1)^{n+1}}{(2n+1)!}\eps^{2n} \eta_i^{2n+1}(\partial_x^{2}+\partial_y^{2})^n (\partial_x u_{i \, 0}+\partial_y v_{i \, 0})\, , 
    \end{split}
    \end{equation}
    where the nondimensional thicknesses are now
    \begin{equation}\label{nondth}
        \eta_1 = \frac{h_1}{{h}}-\alpha \zeta, \quad \eta_2 = \frac{h_2}{{h}}+\alpha \zeta\, .
    \end{equation}
    For the next step we need the following relations 
     \begin{equation}\label{uovotoutilde}
    \begin{split}
        u_{i\, 0}(x,y) &=  \wit{u_i} + \frac{\eps^2}{2}\eta_i^2 \Delta \wit{u_i} +O(\eps^4) \,,\\ 
         v_{i\, 0}(x,y) &=  \wit{v_i} + \frac{\eps^2}{2}\eta_i^2 \Delta \wit{v_i} +O(\eps^4)  \, .
        \end{split}
    \end{equation}
    \medskip
In general, the inversion of such near identity operators leads to the asymptotic formula, 
    \begin{equation}\label{diffop_inv}
         ({\bf 1}+\epsilon^2 {\bf D})^{-1}= {\bf 1}-\epsilon^2 {\bf{D}}+O(\eps^4) . 
    \end{equation}
In the following section we will also need the asymptotic expansions of the layer-averaged horizontal velocities, $\bar{u}_i(x,y),\, \bar{v}_i(x,y)$ defined as 
    \begin{equation}
        \bar{u}_i(x,y):= \frac{1}{\eta_i}\int_{[\eta_i]} u_i(x,y,z)dz\,, \quad \bar{v}_i(x,y):=\frac{1}{\eta_i}\int_{[\eta_i]} v_i(x,y,z)dz\,\,
    \end{equation}
where definite integral notation stands for integration over $z$ in each layer $i=1,2$.   
By the Taylor expansions \eqref{taylor_u_eps}, these are 
    \begin{equation}\label{ubar}
    \begin{split}
        \bar{u}_i(x,y)&= \sum_{n=0}^{\infty}  \frac{(-1)^n}{(2n+1)!}\eps^{2n}\eta_i^{2n}(\partial_x^{2}+\partial_y^{2})^n u_{i \, 0}\, , \\ \bar{v}_i(x,y)&=\sum_{n=0}^{\infty}  \frac{(-1)^n}{(2n+1)!}\eps^{2n}\eta_i^{2n}(\partial_x^{2}+\partial_y^{2})^n v_{i \, 0}\, ,
        \end{split}
    \end{equation}
which can be used to provide their expression in terms of the interfacial variables $\wit{u}_i, \wit{v}_i$. 

Indeed, 
at the lowest order we have
\begin{equation} 
\begin{split}
u_2(x,y,z)&=u_{2\,0} -\frac{\eps^2}{2} \Big(z+\frac{h_2}{{h}}\Big)^2 \Delta u_{2\,0} +O(\eps^4)\, ,\\
v_2(x,y,z)&=v_{2\,0} -\frac{\eps^2}{2} \Big(z+\frac{h_2}{{h}}\Big)^2 \Delta v_{2\,0} +O(\eps^4)\, , 
\end{split}
\end{equation}
whence, by~(\ref{taylor_exp_tilde}) we get 
\begin{equation} 
\begin{split}
u_2(x,y,z)&=\wit{u}_2+\frac{\eps^2}{2}\left(\eta_2^2-(z+\frac{h_2}{{h}})^2\right)\Delta\wit{u}_{2} +O(\eps^4)\,, \\
v_2 (x,y,z) &=\wit{v}_2+\frac{\eps^2}{2}\left(\eta_2^2-(z+\frac{h_2}{{h}})^2\right)\Delta\wit{v}_{2} +O(\eps^4)\, ,
\end{split}
\end{equation}
and, for the vertical velocity,
\begin{equation}
w_2(x,y,z)=-\epsilon \left(z+\frac{h_2}{{h}}\right) \nabla \cdot \begin{pmatrix}
    \wit{u}_{2}\\
    \wit{v}_{2}
\end{pmatrix} +O(\eps^3)\, ,
\label{vertvel2}
\end{equation}
and similar scaling relations holding for the upper layer  velocities.
For the layer averaged velocities, the above relations imply
\begin{equation}\label{tildetobar2}\begin{split}
\ou_i&=\wit{u}_i+\frac{\eps^2}{3}\frac{\eta_i^2}{{h^2}}\Delta \wit{u_i}+O(\eps^4)=\wit{u}_i+ \frac{\eps^2}{3}\frac{h_i^2}{{h^2}}\Delta \wit{u_i} +O(\alpha\eps^2,\eps^4), \\ 
\bar{v}_i&=\wit{v}_i+\frac{\eps^2}{3}\frac{h_i^2}{h^2}\Delta \wit{v_i}+O(\eps^4)\ = \wit{v}_i+\frac{\eps^2}{3}\frac{h_i^2}{h^2}\Delta \wit{v_i}+O(\alpha\eps^2,\eps^4)\, .
\end{split}
\end{equation}
Such asymptotic formulas will be widely used in what follows.

   \subsection{The energy }
    Our next task is to write the explicit form of the energy in terms of horizontal interfacial 
    velocities~$(\wit{u}_i, \wit{v}_i)$. While we can do this at various levels of approximations, for the time being we consider terms of order $O(\eps^2)$, discard terms of order $O(\eps^4)$ or higher, and retain all terms of every order in the parameter $\alpha$. 
All the asymptotic manipulations are needed solely for the kinetic energy,  the potential energy being  straightforwardly computed.  

Let us illustrate the above strategy with the lower fluid. Its kinetic energy density reads
\begin{equation}
\label{T2-0} 
T_2=\frac{\rho_2}{2}\, \int_{-{h_2}/{{h}}} ^{\alpha\zeta}(u_2^2+ v_2^2+ w_2^2)\, {h}\,dz\, ,
\end{equation}
the dimensional factor ${h}$ coming from the scaling of {the physical vertical coordinate} $z$.
This leads to 
\begin{equation}
\label{T2}
\begin{split} 
T_2&=\frac{{h}\, \rho_2}{2}\int_{-\frac{h_2}{{h}}}^{\alpha\zeta} 
\left\{ 
\wit{u}_2^2+\wit{v}_2^2+\epsilon^2
\left[ 
2
\left( 
\wit{u}_2\Delta \wit{u}_{2}+\wit{v}_2\Delta \wit{v}_{2}
\right)
\left(
\eta_2^2-
\left(z+\frac{h_2}{{h}}
\right)^2
\right)
+ \right.\right.\\ &\left.\left.
+\left(
 \wit{u}_{2,x}+\wit{v}_{2,y}
\right)^2
\,
\left(
z+\frac{h_2}{{h}}
\right)^2
\right]
+O(\eps^4)
\right\} 
dz \\& =
\frac{{h}\, \rho_2}{2}\left[\eta_2(\wit{u}_2^2+ \wit{v}_2^2)+\frac{2\eps^2}{3} \eta_2^3 
\left(
\wit{u}_2\Delta\wit{u}_{2}+\wit{v}_2\Delta\wit{v}_{2}+{\frac{1}{2}}
\left(
 \wit{u}_{2\,x}+
    \wit{v}_{2\,y}
\right)^2
\right) 
\right]+O(\eps^4) \, .
\end{split}
\end{equation}
By the same arguments we obtain the contribution to the total kinetic energy density of the upper fluid as
\begin{equation}\label{T1}
T_1= \frac{{h}\, \rho_1}{2}\left[\eta_1(\wit{u}_1^2+ \wit{v}_1^2)+\frac{2\eps^2}{3} \eta_1^3 
\left(
\wit{u}_1\Delta\wit{u}_{1}+\wit{v}_1\Delta\wit{v}_{1}+{\frac{1}{2}}
\left(
  \wit{u}_{1\,x}+
    \wit{v}_{1\,y}
\right)^2
\right) 
\right]+O(\eps^4) \, .
\end{equation}
\subsection{The constraint}\label{DN_op}
As well known (see, e.g. \cite{c_lever97}), in three dimensions  the kinematic boundary conditions at the interface can be expressed as 
\begin{equation}\label{bd_cond}
    \nabla \cdot \left(\eta_1(x,y)\begin{pmatrix}
        \bar{u}_1(x,y)\\
        \bar{v}_1(x,y)
    \end{pmatrix}+\eta_2(x,y)\begin{pmatrix}
        \bar{u}_2(x,y)\\
        \bar{v}_2(x,y)
    \end{pmatrix}\right) = 0\, ,
\end{equation}
where $(\bar{u}_i, \bar{v}_i)$ are 
the layer-averaged horizontal velocities. This implies the existence of a ``mass-stream"-function ${\mu}(x,y)$ such that
\begin{equation}\label{costr_3d}
    \eta_1(x,y)\begin{pmatrix}
        \bar{u}_1(x,y)\\
        \bar{v}_1(x,y)
\end{pmatrix}+\eta_2(x,y)\begin{pmatrix}
        \bar{u}_2(x,y)\\
        \bar{v}_2(x,y)\\
    \end{pmatrix} = \begin{pmatrix}\mu_y\\-\mu_x\end{pmatrix}.
\end{equation}
Since in general the mass-stream function cannot be assumed to vanish, 
a careful asymptotic analysis is needed to arrive at a workable form of the constraint (\ref{costr_3d}), by fixing suitable scaling relations between the long-wave small parameter $\epsilon$ and the non-linearity small parameter $\alpha$.

Looking at the $3D$ velocity potentials $\Phi_i(x,y,z)$ and their interfacial values
\begin{equation}\label{Ztrace}
\wit{\Phi}_i(x,y)=\Phi_i(x,y,z)\big\vert_{z=\zeta(x,y)}\, , 
\end{equation}
we have, e.g., 
\begin{equation}\label{gradpsi}
\wit{\Phi}_{1,x}=\partial_x(\Phi_1(x,y,z))\big\vert_{z=\zeta(x,y)}+\partial_z(\Phi_1(x,y,z))\big\vert_{z=\zeta(x,y)}\zeta_x=\wit{u}_1+\wit{w}_1\zeta_x.
\end{equation}
Since the slope of the normalized interface is small and scales as $O(\alpha\, \epsilon)$, and 
by the Taylor expansion \eqref{taylor_u_eps} the vertical velocities are of order $\epsilon$, in the WNL expansion we get $\wit{u}_1=\wit{\Phi}_{1,x}+O(\alpha \eps^2)$. 
By the same argument we have
\begin{equation}\label{utildepot}
\begin{pmatrix}
\wit{u}_i\\
\wit{v}_i
\end{pmatrix}=\nabla \wit{\Phi}_i+O(\alpha\eps^2).
\end{equation}
By using  relation \eqref{tildetobar2} between averaged and interface velocities together with \eqref{utildepot} above, we can write the mass-balance equations~\eqref{costr_3d} as the system
\begin{equation}\label{co_3Dpsi}
\left\{
\begin{array}{l}\medskip
\eta_1\wit{\Phi}_{1\,x}+\eta_2\wit{\Phi}_{2\,x}+\dsl{\frac{\epsilon^2}{3\, {h}^3}}\left(h_1^3 (\wit{\Phi}_{1\, xxx}+\wit{\Phi}_{1\, yyx})+h_2^3 (\wit{\Phi}_{2\, xxx}+\wit{\Phi}_{2\, yyx})\right)=\mu_y+O(\eps^4, \alpha \eps^2)\\
\eta_1\wit{\Phi}_{1\,y}+\eta_2\wit{\Phi}_{2\,y}+\dsl{\frac{\epsilon^2}{3\, {h}^3}}\left(h_1^3 (\wit{\Phi}_{1\, xxy}+\wit{\Phi}_{1\, yyy})+h_2^3 (\wit{\Phi}_{2\, xxy}+\wit{\Phi}_{2\, yyy}) \right)=-\mu_x+O(\eps^4,\alpha\eps^2)\,.
\end{array}
\right.
\end{equation}
Cross differentiating and taking the difference we get that the mass-stream function satisfies the Poisson equation
\begin{equation}\label{constrfin}
\begin{split}
\Delta\mu=&\alpha (-\zeta_{y}\wit{\Phi}_{1\,x}+\zeta_{x}\wit{\Phi}_{1\,y}+\zeta_{y}\wit{\Phi}_{2\,x}-\zeta_{x}\wit{\Phi}_{2,y})+O(\eps^4,\alpha\eps^2)\\  =&
-  \alpha(\nabla \zeta \times \nabla (\wit\Phi_2-\wit\Phi_1))+O(\eps^4,\alpha\eps^2)\,, 
\end{split}
\end{equation}
where we considered the small amplitude parameter $\alpha$ entering the definition of the normalized interface $\zeta$ (see \eqref{nondth}). By restricting the model in the WNL asymptotics  $ \alpha \ll \epsilon$, i.e., considering terms of order $\alpha, \epsilon^2$ and ignoring terms of order $\alpha^2, \alpha \epsilon^2, \epsilon^4$,
we can trade the general constraint equation \eqref{co_3Dpsi} for the WNL constraint system:
\begin{equation}\label{cofin_3d}
\left\{
\begin{array}{l}\medskip
\eta_1\wit{u}_1+\eta_2\wit{u}_2+\dsl{\frac{\epsilon^2}{3\, {h}^3}}\left(h_1^3 \Delta \wit{u}_1+h_2^3 \Delta \wit{u}_2\right)=
- \alpha \, \partial_y \Delta^{-1}\left(\nabla \zeta \times \begin{pmatrix}
\wit{u}_2-\wit{u}_1\\
\wit{v}_2-\wit{v}_1
\end{pmatrix}\right) +O(\eps^4,\alpha\eps^2) \\
\eta_1\wit{v}_1+\eta_2\wit{v}_2+\dsl{\frac{\epsilon^2}{3\, {h}^3}}\left(h_1^3 \Delta \wit{v}_1+h_2^3 \Delta \wit{v}_2\right)=
\alpha \, \partial_x \Delta^{-1}\left(\nabla \zeta \times \begin{pmatrix}
\wit{u}_2-\wit{u}_1\\
\wit{v}_2-\wit{v}_1
\end{pmatrix}\right)+O(\eps^4,\alpha\eps^2) \,.
\end{array}
\right.
\end{equation}
This constraint can be manipulated to give explicit relations between the velocities of the two layers, better written in the vector form
\begin{equation} \label{u2u1}
\begin{split}
\wit{\mb{u}}_2 = \dsl{-\frac{h_1}{h_2}\wit{\mb{u}}_1} &\dsl{-\frac{\eps^2}{3} \frac{h_1}{h_2}\,{h_1^2-h_2^2 \over h^2}
\Delta\wit{\mb{u}}_1+}\\ \medskip&\dsl{+\alpha  h\, \frac{h_1}{h_2} 
\left(
\frac{1}{h_1}+\frac{1}{h_2}
\right)\left(\zeta\,\wit{\mb{u}}_1+  \, \bm{\nabla} \times\Delta^{-1}(\bm{\nabla} \zeta \times \wit{\mb{u}}_1)\right)} +O(\alpha^2,\eps^4,\alpha\eps^2)\,,
\end{split}
\end{equation}
where $\wit{\mb{u}}_i=\left( \begin{array}{l} \wit{u}_i\\ \wit{v_i}\\
0\end{array}\right)$.
\begin{rem} Contrary to the $2$-dimensional case, relation \eqref{u2u1} introduces non-localities in the reduced theory, which  parallel the ones obtained in the Dirichlet-Neumann (DN) operator approach of \cite{BSL08} for $3$-dimensional $2$-layer configurations with a rigid lid.
As we will explicitly show in the next section, although the constraint \eqref{u2u1} is non-local, the WNL approximation yields a local expression for the Hamiltonian in the canonical variables $(\wit \sigma, \wit \tau)$ prescribed by  the geometric reduction approach of Section \ref{sectHamred}. The comparison of our setting to that  of~\cite{BSL08}  is presented in Appendix \ref{App C}. 
\end{rem}

\subsection{WNL case: the energy in Darboux coordinates}
The final aim of this subsection is to express the energy of the system in terms of the free variables $(\zeta, \wit{\sigma},\wit{\tau})$ introduced in Section \ref{sectHamred}, with the assumption  that the two small parameters in the asymptotics are linked so  that $\alpha$ scales as $O(\eps^2)$, and hence terms of order $O(\alpha^2, \alpha \epsilon^2, \eps^4)$ can be neglected.  This will be done in a few steps.

The first one consists in noticing that in  this asymptotics the  kinetic energy densities  \eqref{T1}, \eqref{T2} preliminarily reduce to
\begin{equation}\label{TsumWNL}
    T_{{\rm WNL}\,i} = \frac{{h}\, \rho_i}{2}\left[\eta_i(\wit{u}_i^2+ \wit{v}_i^2)+\frac{\eps^2}{3} \frac{h_i^3}{{h}^3} 
    \left(2(
    \wit{u}_i\Delta\wit{u}_{i}+\wit{v}_i\Delta\wit{v}_{i})+
\left(
      \wit{u}_{i\,x}+
    \wit{v}_{i\,y}
\right)^2
\right) 
\right]+O( \alpha\eps^2, \eps^4) \, .
\end{equation}
By using the (vector) Green identity
\begin{equation}\label{VGI}
\mb{X} \cdot \Delta\mb{X}+(\boldsymbol{\nabla}\cdot\mb{X})^2=-(|\boldsymbol{\nabla}\times\mb{X}|)^2+\boldsymbol{\nabla}\cdot \left( \mb{X}\times (\boldsymbol{\nabla}\times \mb{X})+\mb{X}\ (\boldsymbol{\nabla}\cdot\mb{X})\right) \,,
\end{equation} 
we rewrite \eqref{TsumWNL} as
\begin{equation}\label{TsumWNL-N}
  T_{{\rm WNL}\,i} = \frac{{h}\, \rho_i}{2}\left[\eta_i(\wit{u}_i^2+ \wit{v}_i^2)+\frac{\eps^2}{3} \frac{h_i^3}{{h}^3} 
    \left(
    \wit{u}_i\Delta\wit{u}_{i}+\wit{v}_i\Delta\wit{v}_{i}-
\left(\wit{v}_{i\,x}-
      \wit{u}_{i\,y}   
\right)^2+T_{{\rm div}\, i}\right)
\right]+O(\alpha \eps^2, \eps^4) \, , 
\end{equation}
where in  $T_{{\rm div}\, i}$ we collected the total divergence terms coming from the last summand in~\eqref{VGI}  with $\mb{X}=\left(\wit{u_i}, \wit{v_i}, 0\right)$.

Next, we observe that the WNL approximation simplifies the effective link between $(\wit{\sigma},\wit{\tau})$ and the velocities. Indeed, 
\eqref{Sigmasheet} yields
\begin{equation}\label{uvtosigmaWNL}
\begin{pmatrix}
\wit{\sigma}\\
\wit{\tau}
\end{pmatrix}=
\begin{pmatrix}
\rho_2 \wit{v}_2-\rho_1 \wit{v}_1+\zeta_y(\rho_2 \wit{w}_2-\rho_1\wit{w}_1)\\
-\left(\rho_2 \wit{u}_2-\rho_1 \wit{u}_1+\zeta_x(\rho_2 \wit{w}_2-\rho_1\wit{w}_1)\right)\end{pmatrix}
=
\begin{pmatrix}
\rho_2 \wit{v}_2-\rho_1\wit{v}_1\\
-\rho_2 \wit{u}_2+\rho_1\wit{u}_1
\end{pmatrix}+O(\alpha\eps^2)\, , 
\end{equation}
since terms like $\zeta_y(\rho_2 \wit{w}_2-\rho_1\wit{w}_1$) scale as $O(\alpha\eps^2)$ thanks to 
(\ref{nondth}) and (\ref{vertvel2}). Moreover, 
combining~\eqref{uvtosigmaWNL} and the constraint~\eqref{u2u1} we arrive at the relations
\begin{equation}
\label{u2u1tau}
\left\{\begin{array}{l}\medskip
\wit{u}_1=\dfrac{h_{2}}{\rho_2 h_1+\rho_1 h_2}\ \wit{\tau}-\dfrac{ \epsilon^{2}}{3}  \dfrac{\rho_2h_1h_2 (h_{1}-h_{2})}{h (\rho_2 h_1+\rho_1 h_2)^2} \Delta \wit{\tau}+\alpha\dfrac{\rho_2 h^2}{(\rho_2 h_1+\rho_1 h_2)^2}\  \left(\zeta\wit{\tau} +\mathcal{N }_1 \right) +O(\alpha^2,\alpha\eps^2, \eps^4)\\
\wit{u}_2=-\dfrac{h_{1}}{\rho_2 h_1+\rho_1 h_2}\ \wit{\tau}-\dfrac{ \epsilon^{2}}{3}  \dfrac{\rho_1 h_1h_2 (h_{1}-h_{2})}{h (\rho_2 h_1+\rho_1 h_2)^2} \ \Delta\wit{\tau}+\alpha\dfrac{\rho_1 h^2}{(\rho_2 h_1+\rho_1 h_2)^2}\ \left(\zeta\wit{\tau} + \mathcal{N}_1\right) +O(\alpha^2,\alpha\eps^2,\eps^4)\, , 
\end{array}
\right.
\end{equation}
and
\begin{equation}
\label{v2v1sigma}
\left\{\begin{array}{l}\medskip
\wit{v}_1=-\dfrac{h_{2}}{\rho_2 h_1+\rho_1 h_2}\ \wit{\sigma}+\dfrac{ \epsilon^{2}}{3}  \dfrac{\rho_2h_1h_2 (h_{1}-h_{2})}{h (\rho_2 h_1+\rho_1 h_2)^2} \ \Delta\wit{\sigma}-\alpha\dfrac{\rho_2 h^2}{(\rho_2 h_1+\rho_1 h_2)^2}\ \left(\zeta\wit{\sigma} + \mathcal{N}_2\right) +O(\alpha^2,\alpha\eps^2, \eps^4)\\
\wit{v}_2=\dfrac{h_{1}}{\rho_2 h_1+\rho_1 h_2}\ \wit{\sigma}+\dfrac{ \epsilon^{2}}{3}  \dfrac{\rho_1 h_1h_2 (h_{1}-h_{2})}{h (\rho_2 h_1+\rho_1 h_2)^2} \ \Delta\wit{\sigma}-\alpha\dfrac{\rho_1 h^2}{(\rho_2 h_1+\rho_1 h_2)^2}\ \left(\zeta\wit{\sigma}+ \mathcal{N}_2 \right) +O(\alpha^2,\alpha\eps^2,\eps^4)\, , 
\end{array}
\right.
\end{equation}
where $\mathcal{N }_i$ is the i-th component of the  non-local vector
\begin{equation}
\bm{\mathcal{N}} = \bm{\nabla} \times \Delta^{-1}\left(\bm{\nabla} \zeta \times \begin{pmatrix}
\wit{\tau}\\
\wit{\sigma}\\
0
\end{pmatrix}\right)\,.
\end{equation}
At the leading order, the previous relations yield
\begin{equation}
\label{O1rel}
\left\{
\begin{array}{l}
\medskip
\wit{u}_{1} = \dfrac{h_{2} \wit{\tau}}{\rho_2h_1+\rho_1 h_2} +O(\alpha,\eps^2)\\ 
\wit{u}_{2} = -\dfrac{h_{1} \wit{\tau}}{\rho_2h_1+\rho_1 h_2} +O(\alpha,\eps^2)
\end{array}
\right. , \quad 
\left\{
\begin{array}{l}
\medskip
\wit{v}_{1}
 = -\dfrac{h_{2} \wit{\sigma}}{\rho_2h_1+\rho_1 h_2}+O(\alpha,\eps^2)\\
\wit{v}_{2} = \dfrac{h_{1} \wit{\sigma}}{\rho_2h_1+\rho_1 h_2}+O(\alpha,\eps^2)
\end{array}
\right.\, .
\end{equation}
Let us first consider the sum of the $O(\eps^2)$  terms $T_{{\rm WNL}\,i}^{(2)}$ in 
\eqref{TsumWNL-N} (disregarding  the total divergence terms $T_{{\rm div}\ i}$).
Substituting the $O(1)$ relations ~\eqref{O1rel} 
we obtain 
\begin{equation}\label{Tnuovo2}
\begin{split}
\sum_{i=1}^2T_{{\rm WNL}\,i}^{(2)}=&
\dfrac12\dfrac{\eps^2}{3\, h^2}
\dfrac{(\rho_1 h_1+\rho_2 h_2) h_1^2h_2^2}
{( \rho_2 h_1+\rho_1 h_2)^2}
\left(\wit{\sigma}\Delta\wit{\sigma}+\wit{\tau}\Delta\wit{\tau}-(\wit{\sigma}_x+\wit{\tau}_y)^2\right)+O(\eps^4,\alpha\eps^2)
\\&
=\dfrac12\dfrac{\eps^2}{3\, h^2}
\dfrac{(\rho_1 h_1+\rho_2 h_2)h_1^2h_2^2}
{( \rho_2 h_1+\rho_1 h_2)^2}
\left(\wit{\sigma}\Delta\wit{\sigma}+\wit{\tau}\Delta\wit{\tau}\right)+O(\eps^4,\alpha\eps^2)\, , 
\end{split}
\end{equation}
owing to the solenoidal property~\eqref{relazdiv}.

Let us now consider the remaining terms in \eqref{TsumWNL-N} that add up to 
\begin{equation}\label{TWNL-0}
\sum_{i=1}^{2} T_{{\rm WNL}\,i}^{(0)}=\frac{h}{2}\left(\rho_1\eta_1(\wit{u}_1^2+\wit{v}_1^2)+\rho_2\eta_2(\wit{u}_2^2+
\wit{v}_2^2)\right) \, .
\end{equation}
 Recalling that  $\eta_i=\dfrac{h_i}{h}+(-1)^i \alpha\zeta$  and expanding in $\alpha$ we get
\begin{equation}\label{T0alpha}
\sum_{i=1}^{2} T_{{\rm WNL}\,i}^{(0)}=\frac{1}{2}\left( \rho_1{h_1}(\wit{u}_1^2+\wit{v}_1^2)+\rho_2{h_2}(\wit{u}_2^2+\wit{v}_2^2)\right)+\alpha\frac{h\zeta}{2}\left(\rho_2(\wit{u}_2^2+\wit{v}_2^2)-\rho_1(\wit{u}_1^2+\wit{v}_1^2)\right)\, .
\end{equation}
In the second summand we can substitute~\eqref{O1rel} to obtain 
\begin{equation}\label{P0alpha1}
\alpha\frac{h\zeta}{2}\left(\rho_2(\wit{u}_2^2+\wit{v}_2^2)-\rho_1(\wit{u}_1^2+\wit{v}_1^2)\right)= 
\alpha\frac{h}{2}\ \dfrac{h_{1}^{2} \rho_{2}-h_{2}^{2} \rho_{1}}{\left(\rho_2h_1+\rho_1 h_2\right)^{2}}\ \zeta\, (\wit{\sigma}^2+\wit{\tau}^2)+O(\alpha^2,\alpha\eps^2)\, .
\end{equation}
To express the first term in~\eqref{T0alpha} in terms of $(\zeta, \wit{\sigma},\wit{\tau})$ we need to use the full relations \eqref{u2u1tau} and \eqref{v2v1sigma}.
The contributions of the $O(\alpha,\eps^2)$ terms in $\frac12 \left( \rho_1{h_1}(\wit{u}_1^2+\wit{v}_1^2)+\rho_2{h_2}(\wit{u}_2^2+\wit{v}_2^2)\right)$ are readily seen to cancel out thanks to the interplay of signs and $h_i,\rho_i$ factors, and so the end result is simply 
\begin{equation}\label{P00ae}
\frac{1}{2}\left( \rho_1{h_1}(\wit{u}_1^2+\wit{v}_1^2)+\rho_2{h_2}(\wit{u}_2^2+\wit{v}_2^2)\right)=\dfrac12\dfrac{h_1h_2}{\rho_2h_1+\rho_1 h_2} \ (\wit{\sigma}^2+\wit{\tau}^2)+O(\alpha^2, \alpha\eps^2,\eps^4)\, .
\end{equation}
Notice that the effective kinetic energy is a local expression in the field variables.
Finally, adding the terms~\eqref{Tnuovo2}. \eqref{P0alpha1} and \eqref{P00ae}, and discarding total divergences we arrive at the final expression of the total effective kinetic energy density  as
\begin{eqnarray}
\begin{split}
    T^{{}^{\rm{eff}}}_{\rm{WNL}} =& 
    \frac{{h}}{2}\left(
   \frac{h_1 h_2}{h(\rho_2h_1+\rho_1h_2)}\left(\wit{\sigma}^2+\wit{\tau}^2\right)+ \alpha\frac{\rho_2h_1^2-\rho_1h_2^2}{(\rho_2h_1+\rho_1h_2)^2} \zeta\left(\wit{\sigma}^2+\wit{\tau}^2\right)+\right.\\  &\left.\frac{\eps^2}{3}\frac{h_1^2h_2^2 (\rho_2 h_2+\rho_1h_1)}{h^3(\rho_2h_1+\rho_1h_2)^2}\left(\wit{\sigma}\Delta\wit{\sigma}+\wit{\tau}\Delta\wit{\tau}\right)
    \right)
    +O(\alpha^2, \alpha\eps^2,\eps^4)\, .
  \end{split}
\end{eqnarray}
The Hamiltonian energy of the system is given by 
\begin{equation}\label{Ham_func}
    \mathcal{H}_{\rm{WNL}}=L^2\int (T^{{}^{\rm{eff}}}_{\rm{WNL}}+U)\, dx \,  dy 
\end{equation}
where $U$ is the potential energy density, 
\begin{equation}
    U = \frac{1}{2}g {h}^2(\rho_2-\rho_1)\zeta^2\,.
\end{equation}
By non-dimensionalizing $(\wit{\sigma},\wit{\tau})$ according to
\begin{equation}\label{fiscalp}
\wit{\sigma}=\wit{\sigma}^*\sqrt{g (\rho_2-\rho_1)(\rho_2 h_1+\rho_1 h_2)}\,,\qquad
\wit{\tau}=\wit{\tau}^*\sqrt{g (\rho_2-\rho_1)(\rho_2 h_1+\rho_1 h_2)}\, ,
\end{equation}
discarding $O(\alpha^2,\alpha\eps^2,\eps^4)$ terms  and dropping  ${}^*$'s gives the Hamiltonian density as 
\begin{eqnarray}\label{Hadimp}
\begin{split}
     H_{\rm{WNL}} =& 
    \frac{{g h^2(\rho_2-\rho_1)}}{2}\left(
   \frac{h_1 h_2}{h^2}\left(\wit{\sigma}^2+\wit{\tau}^2\right)+\zeta^2+ \alpha\frac{\rho_2h_1^2-\rho_1h_2^2}{h\, (\rho_2h_1+\rho_1h_2)} \zeta\left(\wit{\sigma}^2+\wit{\tau}^2\right)+\right.\\  &\left.\frac{\eps^2}{3}\frac{h_1^2h_2^2 (\rho_2 h_2+\rho_1h_1)}{h^4\, (\rho_2h_1+\rho_1h_2)}\left(\wit{\sigma}\Delta\wit{\sigma}+\wit{\tau}\Delta\wit{\tau}\right)
   \right) \, .
  \end{split}
\end{eqnarray} 
Setting hereafter  
\begin{equation}\label{const_expr}
        A = \frac{h_1h_2}{h^2}\,, \quad B =\frac{\rho_2h_1^2 -\rho_1 h_2^2}{h\, (\rho_2h_1 + \rho_1h_2)}\,, \quad \kappa = {{\frac{h_1^2h_2^2 (\rho_1 h_1+ \rho_2h_2)}{h^4\, (\rho_2h_1+\rho_1h_2)}}}\, ,
\end{equation}
and dividing the Hamiltonian density by the factor $D=g h^2(\rho_2-\rho_1)$ (which is tantamount to a time-rescaling), 
the variational derivatives of~(\ref{Hadimp})  can now be computed as
\begin{equation}
    \begin{split}
   \frac{\delta\mathcal{H}}{\delta \zeta}  &=  \zeta + \frac{\alpha}{2}B(\wit{\sigma}^2 + \wit{\tau}^2)\, ,
\\   
   \frac{\delta\mathcal{H}}{\delta \wit{\sigma}}  &= A\wit{\sigma} + \alpha B\zeta\, \wit{\sigma}+\frac{\eps^2}{3}\kappa\Delta \wit{\sigma}\, ,
\\   
   \frac{\delta\mathcal{H}}{\delta \wit{\tau}}  &=  A\wit{\tau} + \alpha B\zeta\, \wit{\tau}+\frac{\eps^2}{3}\kappa\Delta \wit{\tau}\,, 
    \end{split}
    \label{gradH}
\end{equation}
where $\mathcal{H}=\frac{1}{D}\int H_{\rm{WNL}} \,  {d}x\, {d}y$.
Equations of  motion can be found using the Hamiltonian structure \eqref{Pred} as
\begin{equation}
\label{pars}
    \begin{pmatrix}\medskip
         \zeta_t\\\medskip
        \wit{\sigma}_{t}\\
        \wit{\tau}_{t}
    \end{pmatrix} =  \begin{pmatrix}\medskip
        0&-\partial_y&\hspace{0.3cm}\partial_x\\\medskip
        -\partial_y&\hspace{0.3cm}0&\hspace{0.3cm}0\\
        \hspace{0.3cm}\partial_x&\hspace{0.3cm}0&\hspace{0.3cm}0
    \end{pmatrix}\begin{pmatrix}\medskip
        \dsl{\frac{\delta\mathcal{H}}{\delta \zeta}}\\\medskip
          \dsl{\frac{\delta\mathcal{H}}{\delta \wit{\sigma}}}\\
   \dsl{ \frac{\delta\mathcal{H}}{\delta \wit{\tau}}}
    \end{pmatrix}\, ,
\end{equation}
to obtain, explicitly, 
\begin{equation}\label{eqn_WNL}
    \begin{cases}
        &\zeta_t = A(\wit{\tau}_{x}-\wit{\sigma}_{y})+\alpha B((\zeta\wit{\tau})_x-(\zeta\wit{\sigma})_y)+\frac{\eps^2}{3}\kappa((\Delta\wit{\tau})_x-(\Delta \wit{\sigma})_y)\\
        &\wit{\sigma}_{t} = -\, \zeta_y-\alpha {B} (\wit{\sigma}\wit{\sigma}_{y}+\wit{\tau}\wit{\tau}_{y})\\
        &\wit{\tau}_{t} =  \, \zeta_x+\alpha {B}(\wit{\sigma}\wit{\sigma}_{x}+\wit{\tau}\wit{\tau}_{x})\, .
    \end{cases}
\end{equation}  
This system  has to be augmented with the constraint of zero-divergence condition 
\begin{equation} \label{divcondfin}
 \wit{\sigma}_{x}+\wit{\tau}_{y}=0\, .
 \end{equation}
In terms of  the $2D$ vector 
\begin{equation}
\bm{\gamma} =
\begin{pmatrix}
&\hspace{-0.4cm}-\wit{\tau} \\
& \hspace{-0.2cm}\wit{\sigma}
\end{pmatrix}\,,
\label{bmgamma}
\end{equation} 
system \eqref{eqn_WNL} can be written in the more  compact form
\begin{equation}\label{WNL_system_v}
    \begin{cases}
        &\zeta_t = -A\nabla\cdot \textbf{\bm{\gamma}}-\alpha B\nabla \cdot(\zeta {\bm{\gamma}})-\frac{\eps^2}{3}\kappa\nabla\cdot(\Delta \bm{\gamma})\\
        &\bm{\gamma}_t=-g'\, \nabla\zeta- \frac{\alpha}{2}{B} \nabla |\bm{\gamma}|^2
    \end{cases}
\end{equation}
with the additional constraint~(\ref{divcondfin})  becoming the  zero-curl condition
\begin{equation}\label{0curv}
\nabla\times\mb{\gamma}=0\, .
\end{equation}
With the variable $\mb{\gamma}$ thus defined, the Hamiltonian structure \eqref{Pred} assumes the more elegant and compact form
\begin{equation}
   {\cal P}^{red}= \begin{pmatrix}\label{poisson_tensor_v}
        0&-\nabla \cdot\\
        -\nabla & 0
    \end{pmatrix}\, ,
\end{equation}
with Hamiltonian 
\begin{equation}
{\mathcal H}=\frac12 \int
\left(
    (A+\alpha B\zeta)|\bm{\gamma}|^2 
    +\frac{\eps^2}{3}\kappa\Big(\bm{\gamma}\cdot\Delta\bm{\gamma}\Big)+g'\zeta^2
    \right) dx\,dy\, .
    \label{Hamscal}
\end{equation}
Here and in what follows it is useful to insert the non-dimensional parameter $g'$ as a place-holder to highlight the role of the reduced gravity in setting the dynamical evolution.
With these variables, this form of the WNL system is sometimes referred to in the literature as the Kaup-Broer-Kupershmidt-Boussinesq (KBK-Boussinesq) equations~\cite{KS25} (see also~\cite{BSL08,CGK05}) and is the two-dimensional generalization of the weakly nonlinear system in~\cite{CFOPT23}, where the Hamiltonian reduction approach  was used.

The Hamiltonian approach  helps display  a  few conservation laws associated 
with~\eqref{WNL_system_v}.  From the 
structure of the Poisson  tensor,  it  is clear that  the quantities $\int \zeta\, dx\,dy, \, \int \bs{\gamma}\, dx\,dy$ are (trivial)   constants of the motion, as obviously is  the Hamiltonian~\eqref{Hamscal} itself.
Translational invariance of our problem  guarantees  the conservation of the generators  of $x$- and $y$-translations. It is immediate to show  (using the irrotationality condition $\gamma_{1y}=\gamma_{2 x}$) that  they  are given by the  expressions
\begin{equation}
\label{Momentum-conslaws}
\mathcal{M}_1=\int \zeta\gamma_1\, dx\,dy, \quad  \mathcal{M}_2=\int \zeta\gamma_2\, dx\,dy\, .
\end{equation}
The generator of the rotational invariance in the horizontal plane can be  found as well, and it is given by
\begin{equation}\label{angmom}
\mathcal{L} \equiv \int \zeta\, (x\gamma_2-y\gamma_1)\, dx\,dy\, , 
\end{equation}
where we assume that all fields vanish sufficiently fast at spatial infinity. 

\section{Special solutions}\label{specsol}
We first look at the linearization of system~\eqref{WNL_system_v} around constant states which gives  dispersion relations for these equations. Next, we will obtain time-independent solutions of~\eqref{WNL_system_v}, and, last, examine traveling wave solutions. 
\subsection{The dispersion relation}

Linearizing system \eqref{WNL_system_v} around the state $(0,\mb{\gamma_0})\equiv(0, \gamma_{1,0}, \gamma_{2,0})$, 
\begin{equation}
     \begin{cases}
        &\zeta_t +A\nabla\cdot \mb{\gamma} +\alpha B\, \mb{\gamma_0}\cdot\nabla(\zeta) +\frac{\eps^2}{3}\kappa\Delta (\nabla\cdot\mb{\gamma})=0\\
        &\gamma_{1,t}+g'\zeta_x+\alpha {B} (\mb{\gamma_0}\cdot\mb{\gamma})_x=0\\
        &\gamma_{2,t}+g'\zeta_y+\alpha {B} (\mb{\gamma_0}\cdot\mb{\gamma})_y=0
    \end{cases}\, ,
    \label{syssigta}
\end{equation} and seeking solutions by Fourier transform leads to the dispersion relation. 
Thus, by setting
\begin{equation}
        \zeta(x,y,t) = \hat{\zeta} e^{i(\boldsymbol{k}\cdot \bss{x}-\omega t)},\quad        
        \gamma_1(x,y,t) = \hat{\gamma}_1 e^{i(\bss{k}\cdot \bss{x}-\omega t)},\quad
        \gamma_2(x,y,t) = \hat{\gamma}_2 e^{i(\bss{k}\cdot \bss{x}-\omega t)},
        \label{linwnl}
\end{equation}
where $\bm{k}=(k,\ell),\,  \mb{x}=(x,y), \, |\mb{k}|^2=k^2+\ell^2$ and the hatted quantities are constants, 
the left hand side of system~\eqref{syssigta} becomes, after simple manipulations, the matrix multiplication 
\begin{equation}
\left(
\begin{array}{ccc}\medskip
\omega-\alpha B\, \mb{k}\cdot\mb{\gamma_0}  &
    k \left( \frac{ \epsilon^2}{3} \kappa  |\mb{k}|^2-A \right)  &
\ell \left( \frac{\epsilon^2}{3} \kappa  |\mb{k}|^2-A \right)  \\ \medskip
 -g'\,k&\omega-
\alpha B\, k\gamma_{1,0}  &
-\alpha B  \, k \gamma_{2,0}\\ \medskip
-g'\, \ell  &
- \alpha B\,  \ell  \gamma_{1,0} &
 \omega- \alpha B\,  \ell  \gamma_{2,0}
\end{array}\right)
\left(\begin{array}{c}\medskip
    \hat{\zeta}\\ \medskip\hat{\gamma}_1\\ \medskip \hat{\gamma}_2
\end{array}\right)
\,.
\end{equation}
Nontrivial solutions of the system require the determinant of this matrix to be zero, which leads to 
\begin{equation}
         \omega^3-2\alpha B\, (\mb{k}\cdot \mb{\gamma_0})\, \omega^2+\alpha^2 B^2( \mb{k}\cdot\mb{\gamma_0})^2\, \omega+ (g'\kappa \frac{\epsilon^2}{3} \, |\mb{k}|^4 -g' A |\mb{k}|^2)\, \omega=0\, . 
\end{equation}
The dispersion relation is therefore
\begin{equation}\label{disprelbad}
    (\omega -\alpha B\, \bss{k}\cdot \bss{\gamma_0})^2 = g'\, |\mb{k}|^2\Big(A-\frac{\epsilon^2}{3}\kappa|\mb{k}|^2\Big)\, ,
\end{equation}
which shows that the linearized system~\eqref{syssigta} is ill-posed, as the RHS of this equation becomes negative
and is of order $O(|\bss{k}|^4)$ 
as $|\bss{k}| \to \infty$. 
As in the two dimensional case~\cite{CFOPT23}, a regularization of system~\eqref{syssigta} can be achieved by the change of dependent variables 
\begin{equation}\label{chang_coord}
\bar{\bss{\gamma}} =\bss{\gamma}+\frac{\epsilon^2}{3}\frac{\kappa}{A}\Delta \bss{\gamma}\, , \quad \text{ i.e., at } O(\eps^2) \quad 
\bss{\gamma}  =  \bar{\bss{\gamma}}-\frac{\epsilon^2}{3}\frac{\kappa}{A}\Delta  \bar{\bss{\gamma}}\,,
\end{equation}
where inverse operator are computed thanks to their near identity asymptotic weights. 

After linearizing the evolution equations around the constant state $\bss{\gamma_0}$ with these new variables, the dispersion relation of the corresponding non purely evolutionary system becomes
\begin{equation}
     (\omega -\alpha B\, \bss{k}\cdot\bss{\gamma_0})^2  = \frac{g'\, (|\mb{k}|^2)A}{1+\frac{\epsilon^2}{3} \frac{\kappa}{A}(|\mb{k}|^2)}\, ,
\end{equation}
which is well-posed for all $\bss{k}$'s.

Note  that the change of coordinates \eqref{chang_coord} between $\mb{\gamma}$ and $\bar{\bss\gamma}$ corresponds to the change of coordinates~\eqref{tildetobar2} between interfacial and layer averaged variables. 
In fact, coming back to dimensional variables, with the definition 
\begin{equation}
    \begin{pmatrix}
    \bar\gamma_1 \\
     \bar\gamma_2
     \end{pmatrix}
     \equiv \begin{pmatrix}
       \rho_2 \bar u_2- \rho_1 \bar u_1 \\ \rho_2\bar{v}_2 -\rho_1\bar v_1
       \end{pmatrix}
\end{equation} 
and~\eqref{tildetobar2}
we have, e.g. for the second component, 
\begin{equation}\label{comp_sigmabar}
\begin{split}
    \gamma_2 &= \rho_2 \wit{v}_2-\rho_1 \wit{v}_1\\
             &= \rho_2\bar v_2-\rho_1\bar v_1-\rho_2\frac{\eps^2}{3}\frac{h_2^2}{{h}^2}\Delta
          \bar v_2+\rho_1\frac{\eps^2}{3}\frac{h_1^2}{{h}^2}\Delta\bar v_1\, .
    \end{split}
\end{equation}
The definition of $\bar{\gamma}_2$, and use of the WNL constraint at leading order,
\begin{equation}
{h_1}\wit{v}_1 = -{h_2}\wit{v}_2 + O(\alpha) \,,
\end{equation}
yields
\begin{equation}\begin{split}
   \bar{\gamma}_2 &= \frac{\rho_2 h_1+\rho_1 h_2}{h_1} \bar v_2+O(\alpha)\\
   \bar {\gamma}_2 &= -\frac{\rho_2 h_1+\rho_1 h_2}{h_2} \bar v_1+O(\alpha)\,.
   \end{split}
\end{equation}
This can be inverted and plugged  into equation~\eqref{comp_sigmabar} (dropping terms of order $\alpha \epsilon^2$) to find
\begin{equation}
    \begin{split}
         \gamma_2 =\bar{\gamma}_2-\frac{\eps^2}{3}\frac{h_2 h_1(\rho_2h_2+\rho_1 h_1)}{h^2(\rho_2 h_1+\rho_1 h_2)}\Delta
         \bar{\gamma}_2\, , 
    \end{split}
\end{equation}
which plainly holds also for non-dimensional variables. The coefficient  of the Laplacian in this relation is
\begin{equation}
    \frac{h_2 h_1(\rho_2h_2+\rho_1 h_1)}{h^2(\rho_2 h_1+\rho_1 h_2)}= \frac{\kappa}{A}\, , 
\end{equation}
which yields~\eqref{chang_coord}. Similar manipulations hold for the first component.
This shows that the regularizing variable $\bar{\bss{\gamma}}$ represents the jump between the averaged momentum density vectors of each fluid.

\begin{rem} 
The regularization~\eqref{chang_coord} is different from the one used in~\cite{BaCh13}, where the change of variables leading to the shear is done through the choice of the bottom and top layer as reference height of the horizontal velocities in the layers. The resulting model is however still ill-posed with respect to a dispersion critical wavenumber depending on the velocity shear.  It has to be remarked that a velocity shear is intrinsic in our choice of the $\mb{\gamma}$-coordinates, which basically represent a momentum shear, and with the regularizing choice~\eqref{chang_coord} all Fourier modes of the linearized WNL equation are stable, avoiding instabilites that would hamper possible numerical approaches to the study of the system. In the $O(\alpha, \eps^2)$ asymptotics, all these models are (asymptotically) equivalent, albeit with different dispersion relations that  become relevant for large wave numbers $|\mb{k}|$, that is, beyond the expected applicability ranges of the WNL approximation.

The variable change \eqref{chang_coord} is not necessarily confined to the analysis of the linearized model. Indeed, substituting~\eqref{chang_coord} in the full WNL equations~\eqref{syssigta} gives rise, still at $O(\alpha, \eps^2)$, to the following non local system 
\begin{equation}
\label{tipoMMCC}
    \begin{cases}
        &\zeta_t = -A\nabla\cdot \textbf{\bm{\bgamma}}-\alpha B\nabla \cdot(\zeta {\bm{\bgamma}})\\
        &\bm{\bgamma}_t-\frac{\eps^2\kappa}{3\, A} \Delta\mb{\bgamma}_t=-g'\, \nabla\zeta- \frac{\alpha}{2}{B} \nabla  |\bm{\bgamma}|^2
    \end{cases} 
\end{equation}
supplemented with the additional constraint $\nabla\times\mb{\bar{\gamma}}=0$.

From the Hamiltonian point of view, the effects of  the variable change~\eqref{chang_coord} are quite subtle and can be described as follows. 
Since the (Fr\'echet) Jacobian of \eqref{chang_coord} is
\begin{equation}\label{jactilde-bar}
\left(
\begin{array}{ccc}
\mb{1}&0&0\\
0&\mb{1}+\frac{\eps^2\,\kappa}{3\, A}\Delta&0\\
0&0&\mb{1}+\frac{\eps^2\,\kappa}{3\, A}\Delta
\end{array}
\right)\, , 
\end{equation}
the Darboux Poisson tensor \eqref{poisson_tensor_v} is transformed  into the non-canonical one
\begin{equation}
\label{Pbar}
\bar{\mathcal{P}}=
-\left(
\begin{array}{ccc}\medskip
0&\partial_x+\frac{\eps^2\,\kappa}{3\, A}\Delta\partial_x &\partial_y+\frac{\eps^2\,\kappa}{3\, A}\Delta\partial_y\\\medskip
\partial_x+\frac{\eps^2\,\kappa}{3\, A}\Delta\partial_x&0&0\\\medskip
\partial_y+\frac{\eps^2\,\kappa}{3\, A}\Delta\partial_y&0&0\\ 
\end{array}
\right)\, , 
\end{equation}
where $\Delta$ is the horizontal Laplacian, we  used the commutativity properties of the constant-coefficients differential operators, and we fully displayed the tensor to avoid cumbersome notations. To obtain system~\eqref{tipoMMCC} within this Hamiltonian formalism, one has to use the full inversion formula to the first of \eqref{chang_coord}, namely 
\begin{equation}\label{nonloc-cc}
{\mb{\gamma}}=(\mb{\mathcal{T}})^{-1}\mb{\bgamma}\quad\text{with }  \mb{\mathcal{T}}=\mb{1}+\frac{\eps^2\,\kappa}{3\, A}\Delta\, .
\end{equation}
Under such a transformation, the Hamiltonian \eqref{Hamscal} transforms into the non-local expression
\begin{equation}
\label{Hamscalbar}
\bar{{\mathcal H}}=\frac12 \int
\left(
    (A+\alpha B\zeta)|{\BcT^{-1}({\bgamma})}|^2 
    -\frac{\eps^2}{3}\kappa\Big(\BcT^{-1}({\bgamma})\cdot\Delta\BcT^{-1}(\bgamma)\Big)+g'\zeta^2
    \right) dx\, dy\, .
\end{equation}
Further study of this system falls outside the size of this paper, and we leave this to future work.
\end{rem}

\subsection{Time-independent solutions}

We now consider time-independent solutions of the WNL motion equations~\eqref{WNL_system_v}.  Dropping the time derivatives, and setting the parameters $\epsilon$ and $\alpha$ to $1$, we get the system
\begin{equation}\label{WNL_system_vsta}
    \begin{cases}\medskip
        &A\nabla\cdot \bm{\gamma}+\
        B\nabla\cdot (\zeta  \bm{\gamma})
        +\dsl{\frac{\kappa}{3}}\Delta (\nabla\cdot\bm{\gamma}) =0 \,, \\
        &g'\nabla\zeta +\dsl{{B \over 2} }\nabla |\bm{\gamma}|^2=0 \, .
    \end{cases}
\end{equation}
The second equation for $\nabla \zeta$   gives the explicit dependence of $\zeta$ on $|\bm\gamma|$ and the parameter $B$, 
\begin{equation}
\zeta=-{B \over 2\, g'} |\bm\gamma|^2 \,,
\label{zetgam}
\end{equation}
under the assumption of zero velocities and interface displacement at infinity.

The coefficient $B$  is not sign definite, and vanishes at the critical ratio 
$h_1/h_2=\sqrt{\rho_1/\rho_2}$, so that elevation or depression of the interface depends on these ``hardware" parameters. 
Indeed, from \eqref{zetgam}, the elevation $\zeta$ of the interface has the opposite sign of that of the parameter $B$ in~(\ref{pars}).
Hence, the role of $B$ makes the interface move away from the deeper fluid layer, at this order of approximation. Note that this attractive effect of the ``nearest" horizontal  boundary dictated by $B$ is curiously  opposite to its counterpart in the traveling wave solutions, well known for the $2D$ case (see also the next section). 
Whether this phenomenon of ``nearest-boundary attraction" for this class of nontrivial stationary solutions in the two-layer case 
is a mere  artifact of our weakly non-linear asymptotics, or it has a deeper physical meaning, is a problem that deserves a more detailed study, but one that falls beyond the aims of the present paper.

The stationary solutions can also be associated with an explicit expression for the interfacial pressure   $P$.  The layer averaged momentum weakly non-linear equation (4.4) of \cite{ChCa96},  
{
specialized to the case of flat bottom
(see, also, eq.\ (2.20) of~\cite{BaCh13}), yields for the lower fluid interface velocity $\wit{\mb{u}}=(\wit{u}_2,\wit{v}_2)$  the evolution
\begin{equation}\label{Cahlam}
\wit{\mb{u}}_t+(\wit{\mb{u}}\cdot \nabla) \wit{\mb{u}}+g (1-\frac{\rho_1}{\rho_2})\nabla\zeta=-\frac{\nabla P}{\rho_2}+ \frac{\rho_1}{\rho_2}( \frac{h_1^2}2\nabla(\nabla\cdot \wit{\mb{u}})_t+h_1\zeta_{tt})  +o(\epsilon^2, \alpha)\, .
\end{equation}
Restoring  the asymptotic small parameter $\alpha\simeq \epsilon^2$ in the quadratic inertia term, and considering stationary  flows,
yields the Bernoulli relation
\begin{equation}
P=-(\rho_2 -\rho_1) g \zeta-{\alpha \over 2}\rho_2 (\wit{u}_2^2+\wit{v}_2^2) \, ,
\end{equation}
since the interface is a streamline.
From relations~\eqref{O1rel} we get 
\begin{equation}
P=-(\rho_2-\rho_1)  g \zeta-\frac{\alpha}{2}\, \frac{h_1^2\rho_2}{(\rho_2 h_1+\rho_1 h_2)^2}|\bm\gamma|^2\,.
\end{equation}
This relation shows that to leading order the pressure is hydrostatic, i.e., follows the interfacial displacement, and the weighted vorticity contributes a weak, $O(\alpha)$, depressive term.}
\begin{rem} Substituting \eqref{zetgam} in the first of \eqref{WNL_system_vsta} we arrive at the following equation for $\bs{\gamma}$
\begin{equation}\label{eqgamdiv}
\nabla\cdot\left(
A\bm{\gamma}\
        -\dsl{\frac{B^2}{2\, g'}} ( |\bm\gamma|^2\bm{\gamma})
        +\dsl{\frac{\kappa}{3}}\Delta \bm{\gamma}
\right)=0 \, . 
\end{equation}
Standing, irrotational  wave solutions with frequency ${A}$  of the defocusing vector Nonlinear  Schr\"odinger equation
\begin{equation}\label{VNLSE}
i\partial_t\bm{\gamma}-\dsl{\frac{B^2}{2\, g'}} ( |\bm\gamma|^2\bm{\gamma})
        +\dsl{\frac{\kappa}{3}}\Delta \bm{\gamma} =0\, 
\end{equation}
satisfying  suitable boundary conditions provide non-trivial solutions to \eqref{eqgamdiv}, as noticed in the recent paper \cite{KS25}. Further study on this connection is ongoing.
\end{rem}

\subsection{Traveling solitary plane wave solutions}\label{TPSW-KBKB}

By using rotational invariance, we can reduce 
the plane-traveling wave ansatz for system~\eqref{WNL_system_v}  to 
\begin{equation}\label{ptw-Ans}
    \zeta(x,y,t) = \zeta(x-ct)\,, \quad \gamma_1(x,y,t) = \gamma_1(x-ct )\,, \quad \gamma_2(x,y,t) = \gamma_2(x-ct )\, .
\end{equation}
Integrating once with respect to $x$ and taking into account the far-field vanishing boundary condition turns the system into 
\begin{equation}\label{pwt2sys}
     \begin{cases}\medskip
        -c \zeta + A \gamma  +  B\,\zeta \gamma +\dsl{\frac{\kappa}{3}}\gamma_{xx}=0\\
        -c\gamma+ g'\zeta+\dsl{\frac{B}{2}}\,\gamma^2=0\\
            \end{cases}
\end{equation}
for the pair $(\zeta, \gamma\equiv\gamma_1)$, since under the plane-traveling wave ansatz \eqref{ptw-Ans} the $y$ component  $\gamma_2$ of the vector $\bm{\gamma}$ vanishes.
This is the traveling wave system associated with the $1+1$ dimensional system of PDEs
\begin{equation}
\label{1D-BK}
\begin{cases}\medskip
&\zeta_t+A \gamma_x  +  B(\zeta \gamma)_x +\frac{\kappa}{3}\gamma_{xxx}=0\\
&\gamma_t+ g'\zeta_x+{B}\gamma\, \gamma_x=0\\
\end{cases}\, , 
\end{equation}
which is the parametric form of the classical $1+1$-dimensional Kaup-Broer-Kupershmidt integrable shallow water equation considered in \cite{CFOPT23}.
Its one-soliton solution can be found in a few simple steps, which we  report here for the sake of completeness. Solving the second algebraic equation of \eqref{pwt2sys}  we get
\begin{equation}\label{gammatozeta}
\zeta=\frac1{g'}\Big( c\gamma-\dsl{\frac{B}{2}}\, \gamma^2\Big)\, , 
\end{equation}
so that we obtain  for $\gamma$ the  ODE
\begin{equation}\label{Neweq}
\frac{2 g' \kappa }{ 3}\ \gamma_{{xx}}  = 
B^{2} \gamma^{3}-3 c \,\gamma^{2} B +2\left( c^{2}-Ag'\right) \gamma
\end{equation}
equivalent to a Newton Second Law of motion for a point-particle of mass $\mu=\dsl{\frac{2 g' \kappa }{ 3}}$ in the  potential energy
\begin{equation}\label{Ueff}
U=\gamma^2\left(
(A g' -c^{2})+B c \gamma-\frac{1}{4} B^{2} \gamma^{2} 
\right) \,.
\end{equation}
The first condition to have a traveling wave is the usual one, stating that soliton travel faster than the linear waves $c_0=\sqrt{g'A}$,  i.e., 
\begin{equation}\label{c2>gA}
c>c_0\, .
\end{equation}
Then we can a priori notice that the  amplitude   of the $\gamma$-peak is given by
\begin{equation}\label{gamax}
\gamma_m = \frac{2}{B} (c - c_0)\, , 
\end{equation}
while the explicit solution expression is 
\begin{equation}\label{gammasol}
\gamma(x,t)=\frac{4 \,\left(c^{2}-c_{0}^{2}\right)}{B}\>\frac{{\mathrm e}^{-\left(x-c t  \right) c_\Delta } }{ \left(2 c \,{\mathrm e}^{-\left(x-c t \right) c_\Delta} +c_{0} {\mathrm e}^{-2 \left(x-c t\right) c_\Delta}+c_{0}\right)}\, , 
\end{equation}
and the corresponding interface elevation  solution is given by
\begin{equation}
\label{intelsol}
\zeta(x,t)= \frac{4 A \left(c^{2}-c_{0}^{2}\right)}{c_{0} B} \>\frac{{\mathrm e}^{-\left(x-c t\right) c_\Delta} \left(2   c_{0}\,{\mathrm e}^{-\left(x-c t\right) c_\Delta}+ c\,{\mathrm e}^{-2\left(x-c t\right) c_\Delta} +c \right) }{ \left(2 c\, \,{\mathrm e}^{-\left(x-c t\right) c_\Delta} +c_{0} {\mathrm e}^{-2\left(x-c t\right) c_\Delta}+c_{0}\right)^{2}}\, ,
\end{equation}
where $c_\Delta=\sqrt{({c^2-c_0^2})/{\mu}}$.
This solution is well known from the theory of the KBK-B system.  The remarkable fact in the $2$-layer setting is that waves are of elevation for $B>0$ and of depression for $B<0$, that is, we have solitary waves of elevation (`bright') when $h_1>\sqrt{\frac{\rho_1}{\rho_2}}\, h_2$, and of depression (`dark')  in the opposite case.

The no-contact condition with the plates of the channel translates into
\begin{equation}\label{notouch1}
\begin{cases}
\zeta_m<\frac{h_1}h\quad \text{for } B>0\, , \\
\zeta_m>-\frac{h_2}h\quad \text{for } B<0\, , 
\end{cases}
\end{equation}
where the maximum $\zeta$ amplitude is computed from \eqref{intelsol} as
$$
\zeta_m=\dsl{\frac{2\, A}{B}\frac{c-c_0}{c_0}}\, .
$$
Condition \eqref{notouch1} together with the characteristic condition \eqref{c2>gA} for the existence of solitary wave solutions  is satisfied for
\begin{equation}
\begin{cases}\medskip
c_0 <c <c_0 \dsl{(1+\frac{B\,h_1}{2A\, h})}\quad  \text{for } B>0\, , \\
c_0 <c <c_0 \dsl{(1-\frac{B\, h_2}{2A\, h})}\quad  \text{for }B<0\, ,\\
\end{cases}
\end{equation}
whence there exists a range of speeds (and corresponding amplitudes) supporting traveling ``line''  solitary wave solutions in the $2$-layer case.
We remark that, notwithstanding the intrinsic limitations associated with our WNL approximation, the expected appearance of bright ($B>0$) and dark ($B<0$) solitary wave is recovered. It should also be noticed that  $B$   is small whenever $h_1/h\sim h_2/h=O(1)$, so that the range of admissible wave speeds is quite limited. Moreover the model does not support heteroclinic orbits and hence limiting conjugate-state-like solutions \cite{ChCa99}.

\section{The KP regime}\label{KP}
In this section we show  how the so called Kadomtsev-Petviashvili equation arises in the process of unidirectionalization of the KBK-Boussinesq system~\eqref{WNL_system_v}, by assuming a slow variation in the $y$-direction and a weak non linearity. 
We will first derive the KP equation directly from the KBK-Boussinesq system, and then equip it with the reduced Hamiltonian structure associated with the asymptotic process with the systematic Dirac reduction technique.

The KP asymptotic regime  is obtained by breaking the rotational symmetry of the KBK-B equations, and assuming that the horizontal independent  original variables have different scaling laws along the two horizontal directions, namely $x = Lx^*$ and $y = L'y^*$, with  ${L}/{L'} = O(\eps)$. In terms of the rescaled variables we are using
we have
\begin{equation}\label{scaling_y}
  y'=\beta y\,  \Rightarrow \partial_y=\beta\partial_{y'} \qquad \hbox{with} \quad \beta = O(\eps)\, .
\end{equation}
As a  notable consequence of this modified $y$-scaling the second component $\gamma_2$ is to be replaced by 
\begin{equation}\label{gamma2prime}
\gamma_2^\prime=\dsl{\frac1{\beta}}\gamma_2\, ,
\end{equation}
owing to the fact that $\bm{\gamma}$ is irrotational, due to its definition~\eqref{bmgamma} and the divergence free condition $\wit{\sigma}_x+\wit{\tau}_y=0$.

Under the WNL assumptions  the KBK-Boussinesq model \eqref{WNL_system_v} becomes
\begin{equation}\label{KBK1}
   \begin{cases}\smallskip
       \zeta_t+A\gamma_{1\,x}+\beta^2 A \gamma_{2\,y^\prime}^\prime+\alpha B(\zeta \gamma_1)_x+\dsl{\frac{\eps^2}{3}}\kappa\gamma_{1\,xxx}=0\\ \smallskip
       \gamma_{1\,t}+g'\zeta_x+\alpha \frac{B}{2}(\gamma_1^2)_x=0\\
       \beta  \gamma^\prime_{2\,t}+\beta g'\zeta_{y^\prime}+\alpha \beta \frac{B}{2} (\gamma_1^2)_{y^\prime}=0\, .
   \end{cases}
\end{equation}
This system has to be augmented by the irrotationality condition that now reads 
   $ \beta\gamma^\prime _{2\,x}=\beta \gamma_{1\,y^\prime}\, .$
Dropping primes from $y$-derivatives and $\gamma_2$ we finally arrive at
\begin{equation}\label{full_mod_beta}
    \begin{cases}
       \zeta_t+A\gamma_{1\,x}+\beta^2 A \gamma_{2\,y}+\alpha B(\zeta \gamma_1)_x+\frac{\eps^2}{3}\kappa\gamma_{1\,xxx}=0\\
       \gamma_{1\,t}+g'\zeta_x+\alpha \frac{B}{2}(\gamma_1^2)_x=0\\
 \gamma_{2\,t}+ g'\zeta_y+\alpha \frac{B}{2} (\gamma_1^2)_y=0\\
      \gamma_{2\,x}=\gamma_{1\,y}\, .
   \end{cases}
\end{equation}
We now proceed with the derivation of the unidirectional model, similarly to the case of the KdV equation in~\cite{Wh2000}, by seeking a relation between the dependent variables that makes the first two equations of system~\eqref{full_mod_beta} identical. The rationale for this is that
the third equation of system~\eqref{full_mod_beta} is simply a compatibility condition on the components of~$\bm{\gamma}$,  satisfied  at the asymptotic order we are considering.

We thus seek a relation (akin to a Riemann invariant) that keeps track of the contribution of weak nonlinearity and dispersion to the right going linearized restricted system, 
\begin{equation}
    \begin{cases}
        \zeta_t +A\gamma_{1\,x}=0\\
        \gamma_{1\,t}+g'\zeta_x=0\, ,
    \end{cases}
\end{equation}
as 
\begin{equation}
      \gamma_1 = \sqrt{\frac{g'}{A}}\zeta+\alpha F(\zeta)+\epsilon^2 G(\zeta)
           +O(\alpha\eps^2,\alpha^2,\eps^4,\beta^2)
\end{equation}
where $F(\zeta)$ and $G(\zeta)$ are  differential polynomials in the variable $\zeta$. Making the first two equations in~\eqref{full_mod_beta} equivalent up to order $O(\beta^2)$ yields 
\begin{equation}
   \gamma_1 = \sqrt{\frac{g'}{A}}\zeta-\frac{\alpha}{4} \frac{B}{A}\sqrt{\frac{g'}{A}}\zeta^2-\frac{\eps^2}{6}\frac{\kappa}{A}\sqrt{\frac{g'}{A}}\zeta_{xx}+O(\beta^2)   \,.
\end{equation}
Next, we account for the $O(\beta^2)$ terms in the system
\begin{equation}
    \begin{cases}
       \zeta_{\,t}+A\gamma_{1\,x}+\beta^2 A\partial_x^{-1}\gamma_{1\,yy}+\alpha B(\zeta \gamma_1)_{x}+\frac{\eps^2}{3}\kappa\gamma_{1\,xxx}=0\\
       \gamma_{1\,t}+g'\zeta_{x}+\alpha \frac{B}{2}(\gamma_1^2)_{x}=0\,.
   \end{cases} 
   \label{interKP} 
\end{equation}
The relation between $\gamma_1$ and $\zeta$ to make the two equations equivalent at order~$O(\beta^2)$ can now be sought as 
\begin{equation}
   \gamma_{1\, }= \sqrt{\frac{g'}{A}}\zeta-\frac{\alpha}{4} \frac{B}{A}\sqrt{\frac{g'}{A}}\zeta^2-\frac{\eps^2}{6}\frac{\kappa}{A}\sqrt{\frac{g'}{A}}\zeta_{xx}+\beta^2K(\zeta)  \, ,
\end{equation}
which, once inserted  in~\eqref{interKP}, results in 
\begin{equation}\label{kp_constr}
   \gamma_{1}= \sqrt{\frac{g'}{A}}\left(\zeta-\frac{\alpha}{4} \frac{B}{A}\, \zeta^2-\frac{\eps^2}{6}\frac{\kappa}{A}\zeta_{xx}-\frac{\beta^2}{2}\partial_x^{-2}\zeta_{ yy}\right) \, .
\end{equation}
The final form of the sought-after unidirectional equation, after taking an $x$-derivative, is thus
\begin{equation}\label{KP_eq}
     \Big(\zeta_{ t}+\sqrt{Ag'}\zeta_{x}+\frac{3}{4}\alpha \sqrt{\frac{g'}{A}}B(\zeta^2)_{x}+\frac{\eps^2}{6}\kappa\sqrt{\frac{g'}{A}}\zeta_{xxx}\Big)_x+\frac{\beta^2}{2}\sqrt{A g'} \, \zeta_{yy} = 0\, ,
\end{equation}
which is a Kadomtsev-Petviashvili (KP-II) equation with the coefficients determined by the hardware parameters $\rho_i, h_i$. Note that coefficient $B$ of the nonlinear term can change sign across the critical height relation
\begin{equation}
{h_1\over h_2}=\sqrt{\rho_1\over \rho_2}\,.
\end{equation}

\subsection{The KP Hamiltonian structure}
The KP bi-Hamiltonian structure can be obtained within our scheme by means of a suitable Dirac Hamiltonian reduction \cite{Dirac}.
By applying the scaling \eqref{scaling_y},  with $\beta \sim\eps\sim\sqrt{\alpha} $ to  the Hamiltonian functional, taking into account~\eqref{gamma2prime}, and  discarding terms of order higher than $\alpha, \eps^2,\beta^2$, we obtain  the asymptotic Hamiltonian in the KP regime 
\begin{equation}\label{H_KP-full}
    \mathcal{H}_{KP} = \frac12\int  \left(A\, \gamma_1^2 +\beta^2 A\gamma_2^2 +\alpha B\zeta \gamma_1^2+\frac{\eps^2}{3}\kappa\gamma_1 \gamma_{1\,xx}+{g'}\zeta^2 \right) dx\, dy\, ,
\end{equation}
whose Hamiltonian equations of motion given by the Poisson operator \eqref{poisson_tensor_v} and  supplemented by the   irrotationality condition $\gamma_{2,x}= \gamma_{1,y}$  are
\begin{equation}
    \begin{cases}
       \zeta_t+A\gamma_{1\,x}+\beta^2 A \gamma_{2\,y}+\alpha B(\zeta \gamma_1)_x+\frac{\eps^2}{3}\kappa\gamma_{1\,xxx}=0\\
       \gamma_{1\,t}+g'\zeta_x+\alpha \frac{B}{2}(\gamma_1^2)_x=0\\
        \gamma_{2\,t}+ g'\zeta_y+\alpha \frac{B}{2} (\gamma_1^2)_y=0 \\
        \gamma_{2\,x}=  \gamma_{1\,y}\,  .
        \end{cases}
\end{equation}

The system thus obtained coincides with the direct asymptotic expansion \eqref{full_mod_beta} of the KBK-B model we presented above, a consequence that comes all the way back from the constant nature of the Hamiltonian operator~(\ref{Pred}) from the first two-layer reduction.
The idea is to consider the irrotationality condition and the unidirectional constraint  in the form given by~\eqref{kp_constr} 
as constraints tying $\zeta$ and $\gamma_1$, and  find the Hamiltonian structure(s) for the KP equation \eqref{KP_eq}  by using the extension to the infinite dimensional case of the  Dirac reduction procedure of the standard Poisson tensor \eqref{poisson_tensor_v}.
To this end,  we recall that the celebrated Dirac reduced bracket~\cite{Dirac} for a finite-dimensional Hamiltonian system with coordinates 
$(y_1, \ldots, y_N)$ subject to $M$ constraints $\Phi_a=0, a=1,\ldots, M \,, M<N$, such that the determinant of the matrix of the Poisson brackets of the constraints $\mathcal{C}_{a,b}=\{\Phi_a, \Phi_b\} $ be invertible, reads
\begin{equation}
\label{DirFD}
\{y_i, y_j\}^D=\{y_i, y_j\}-\sum_{a,b=1}^M \{y_i, \Phi_a\}(\mathcal{C}^{-1})_{ab} \{\Phi_b, y_j\}\, . 
\end{equation}
It is quickly realized that the constraints $\Phi_a$, $a=1,\ldots,M$, are Casimir functions for the Dirac bracket, which implies that it restricts directly to the constrained manifold defined by the equations $\Phi_a=0$, $a=1\ldots,M$. This observation suggests that a suitable procedure
to obtain the Poisson tensor corresponding to the Dirac bracket \eqref{DirFD} can be the following:
\begin{enumerate}
\item Pick a set of coordinates adapted to the constraints, that is, consider the coordinate set $(y_{j_1}, \ldots, y_{_{N-M}}, \Phi_1, \ldots, \Phi_M)$.
\item Write the Poisson tensor $P$ corresponding to the original bracket $\{\cdot, \cdot\}$ in these new coordinates to obtain the new matrix representation
\begin{equation}
\label{PDtilde}
\wit{P}=\left(
\begin{array}{cc}
\{y_k, y_l\}&\{y_i, \Phi_a\}\\
\{\Phi_a,y_\ell\}&\{\Phi_a,\Phi_b\}
\end{array}
\right)
\equiv \left(
\begin{array}{cc}
\mathcal{A}& \mathcal{B}\\
-\mathcal{B}^T&\mathcal{C}
\end{array}\right)\,.
\end{equation}
\item Apply formula \eqref{DirFD} to read the Dirac reduced tensor as
\begin{equation}
\label{DirRP}
\wit{P}^D=\left(
\begin{array}{cc}
\mathcal{A}- \mathcal{B}\cdot\mathcal{C}^{-1}\cdot \mathcal{B}^T&0\\0&0
\end{array}\right)\, , 
\end{equation}
 whose restriction to the constrained manifold, parametrized by the coordinates $(y_{j_1}, \ldots, y_{_{N-M}})$, is given by the $(N-M)\times(N-M)$ upper left block
$\mathcal{A}- \mathcal{B}\cdot\mathcal{C}^{-1}\cdot \mathcal{B}^T$.
\end{enumerate}
Such a procedure can be easily implemented ({\em mutatis mutandis} and with a grain of salt) to the field theoretical model we are herewith studying as follows.
 
We consider the  expression of the constraints~\eqref{kp_constr} as defining two new  coordinates as 
\begin{equation}
\begin{split}
    \Phi_1 &\equiv  \gamma_1 - \sqrt{\frac{g'}{A}}\zeta+\frac{\alpha}{4} \frac{B}{A}\sqrt{\frac{g'}{A}}\zeta^2+\frac{\eps^2}{6}\frac{\kappa}{A}\sqrt{\frac{g'}{A}}\zeta_{xx}+\frac{\beta^2}{2}\sqrt{\frac{g'}{A}}\partial_x^{-2}\zeta_{ yy}\\
    \Phi_2 &\equiv \gamma_{2\,x}- \gamma_{1\,y}\, .
    \end{split}
\end{equation}

As a first step in the  Dirac reduction procedure we express the Poisson structure \eqref{poisson_tensor_v} in the coordinates   $(\zeta, \Phi_1, \Phi_2)$.
The Jacobian operator-valued matrix of the transformation $(\zeta, \gamma_1,\gamma_2)\mapsto (\zeta, \Phi_1, \Phi_2)$ is
\begin{equation}
    J = \begin{pmatrix}
        1 &0&0\\
        \delta_\zeta \Phi_1&1&0\\
        0 & -\partial_y & \partial_x
    \end{pmatrix}\, ,
\end{equation}
where  $\delta_\zeta \Phi_1$ is the Fr\'echet derivative of $\Phi_1$ with respect to $\zeta$,  given by \begin{equation}
    \delta_\zeta \Phi_1  = - \frac{g'}{c_0}+\frac{\alpha}{2}\frac{B}{A}\frac{g'}{c_0}\zeta +\frac{\eps^2}{6}\frac{\kappa}{A}\frac{g'}{c_0} \partial_{xx}+\frac{\beta^2}{2}\frac{g^\prime}{c_0} \,\partial_x^{-2}\partial_y^2\,.
\end{equation}
Here to improve readability we used the relation $c_0=\sqrt{A\, g'}$ tying the speed $c_0$ of the linear waves to the ``mass" and gravity parameters of the problem.

The Poisson tensor \eqref{poisson_tensor_v} is given in terms of the new coordinates $(\zeta, \Phi_1, \Phi_2)$ by 
\begin{equation}
   {\cal P}_{KP} = J {\cal P}J^T = \begin{pmatrix}
       0 &-\partial_x &0\\
       -\partial_x& \mathcal{P}_{22}(\zeta, \partial_x^{-1}, \partial_y)&0\\
       0 & 0 &0
   \end{pmatrix}
\end{equation}
where, explicitly,
\begin{equation}\label{P22}
\begin{split}
  \mathcal{P}_{22} (\zeta, \partial_x^{-1}, \partial_y) =& -(\partial_x  (\delta_\zeta \Phi_1)^T+ (\delta_\zeta \Phi_1)\partial_x ) =\\ &2 \frac{g'}{c_0}\partial_x -\frac{\alpha}{2}\frac{B}{A} \frac{g'}{c_0}(\partial_x \zeta + \zeta \partial_x)-\frac{\eps^2}{3}\frac{\kappa}{A} \frac{g'}{c_0} \partial_{xxx}- \beta^2\frac{g'}{c_0} \partial_y \partial_x^{-1} \partial_y\,.
    \end{split}
\end{equation}
From this expression we have that $\Phi_2$ is already a Casimir of ${\cal P}_{KP}$, so that we need to use a Dirac reduction only on the constraint $\Phi_1$, 
and obtain the reduced Dirac tensor on the constrained manifold defined by $\Phi_1=0\,, \Phi_2=0$, parametrized by the free coordinate $\zeta$,  by the effective formula 
\begin{equation}\label{P_kp_red}
    {\cal P}_{KP}^{D} = \partial_x
   \left(\partial_x  \delta_\zeta \Phi_1+\delta_\zeta \Phi_1\partial_x \right)^{-1}\partial_x\, .
\end{equation}
The operator $\partial_x  (\delta_\zeta \Phi_1)^T+ (\delta_\zeta \Phi_1)\partial_x \equiv  \partial_x  (\delta_\zeta \Phi_1)+ (\delta_\zeta \Phi_1)\partial_x $ turns out to be invertible in the ring of microdifferential operators, with an inverse that
at this asymptotic order is given by 
\begin{equation}\label{inv_canc}
    \left(\partial_x  (\delta_\zeta \Phi_1)+ (\delta_\zeta \Phi_1)\partial_x\right)^{-1} = -\frac{1}{4}\frac{c_0}{g'}\left(2 \partial_x^{-1} + \frac{\alpha}{2}\frac{B}{A}(\partial_x^{-1} \zeta+ \zeta \partial_x^{-1})+\frac{\eps^2}{3}\frac{\kappa}{A}\partial_x+\beta^2 \partial_y \partial_x^{-3}\partial_y\right)\,,
\end{equation}
as can be shown by direct computation. 

By inserting the inverse operator \eqref{inv_canc} in the expression  \eqref{P_kp_red} we finally find
\begin{equation} \label{red_kP}
     {\cal P}_{KP}^{D} = \frac{1}{4}\frac{c_0}{g'}\left(-2 \partial_x-\frac{\alpha}{2}\frac{B}{A}(\zeta \partial_x +\partial_x \zeta)-\frac{\eps^2}{3}\frac{\kappa}{A}\partial_{xxx}-\beta^2 \partial_y \partial_x^{-1}\partial_y\right)\,.
\end{equation}
We remark that,  consistently collecting the three small parameters $\alpha\sim\epsilon^2\sim\beta^2$ and suitably rescaling variables to get rid of the hardware parameters $A$ and $\kappa$, this expression reduces to 
\begin{equation}\label{KP-pencil}
{\cal P}_{KP}^{D}\simeq -2 \partial_x-\alpha\left(\partial_{xxx}+2\zeta\partial_x+\zeta_x + \partial_y \partial_x^{-1}\partial_y\right) \, , 
\end{equation}
which brings it closer to the notation in the literature for the Poisson pair of the KP equation, viewed as evolutionary non-local partial differential equation in $2$-space dimensions (see also \cite{FoSa88}), where $\alpha$ plays the role of (inverse) parameter of the Poisson pencil.

As a final step, we have to determine the asymptotic expansion of the constrained Hamiltonian. 
Enforcing the constraints $\Phi_1=0,\,\Phi_2 =0$ in the expression of the functional \eqref{H_KP-full}  we obtain the preliminary form of the constrained Hamiltonian as 
\begin{equation}
    \mathcal{H}_{KP}^{(c)} = \int \left({g'}\zeta^2+\frac{\alpha}{4}B  \frac{g'^2}{c_0^2} \zeta^3 + \frac{\beta^2}{2}A \frac{g'^2}{c_0^2}\Big((\partial_x^{-1}\partial_y\zeta)^2-\zeta \partial_x^{-1}\partial_y\partial_x^{-1}\partial_y\zeta \Big)\,  \right) dx\, dy\, , 
    \label{prelHc}
\end{equation}
where we used \eqref{kp_constr}, and  $\gamma_2 = \partial_x^{-1}\partial_y \gamma_1$ so that  $ \gamma_2 = \dsl{ \frac{g'}{c_0}}\partial_x^{-1}\partial_y\zeta + O(\alpha, \eps^2, \beta^2)\, $,  and the fact that $ {A\, g'^2}/{c_0^2}\equiv g'$.
   
Thanks to the symmetry of the operator $\partial_x^{-1}\partial_y$, the $O(\beta^2)$-term in \eqref{prelHc}  can be discarded, to conclude that the restricted Hamiltonian is 
    \begin{equation}\label{Kp-restr-ham}\
    \mathcal{H}_{{KP}}^{(c)} = \int \left(g'\zeta^2 +\frac{\alpha}{4}B  \frac{g'^2}{c_0^2} \zeta^3\right) dx\, dy \, .
    \end{equation}
As a final check,  we note that the variational derivative of the constrained Hamiltonian \eqref{Kp-restr-ham} with respect to $\zeta$ is
\begin{equation}
    \frac{\delta\mathcal{H}_{KP}^{(c)}}{\delta \zeta} = 2 g'\zeta +\alpha \frac{3}{4} \frac{g'^2}{c_0^2}B\zeta^2\, .
\end{equation}
By applying the Dirac-reduced Poisson tensor \eqref{red_kP} to  this differential we obtain
\begin{equation}
    \zeta_t +c_0\zeta_x+\alpha \frac{3}{4} B \frac{g'}{c_0}(\zeta^2)_x+\frac{\eps^2}{6}\kappa\frac{g'}{c_0}\zeta_{xxx}+\frac{\beta^2}{2}c_0\partial_x^{-1}\zeta_{yy} = 0\, ,
\end{equation}
which, after taking an $x$-derivative, is exactly the KP-II equation \eqref{KP_eq}.
\subsection{KP traveling waves in the channel}
For the sake of completeness, we shall now derive the solitary wave solution to the KP equation, by means of the traveling-wave ansatz
\begin{equation}
    \zeta(x,y,t) = \zeta(x-ct +qy)\, ,
\end{equation}
where $x+qy=0$ is the equation defining the ``line" traveling wave, where $q$ according to the weak transversal dependence assumption, is small.
The KP equation \eqref{KP_eq} transforms  into 
\begin{equation}
    (c_0-c)\zeta_{xx}+\frac{3}{4}\alpha \frac{g'}{c_0}B(\zeta^2)_{xx}+\frac{\eps^2}{6}\frac{g'}{c_0}\kappa \zeta_{xxxx}+\frac{\beta^2}{2}c_0 q^2\zeta_{xx}=0
\end{equation}
where we kept the definition of  $c_0=\sqrt{A g' }$. 
Setting for simplicity $\alpha = \eps^2 = \beta^2 = 1$ and integrating  we obtain the usual Newton 
zero-energy relation for a point mass nonlinear oscillator  with a cubic potential of the form 
\begin{equation}\label{equa}
    \frac{\kappa g'}{12 c_0} \zeta_x^2 = -\left(\frac{c_0(1+\frac{q^2}{2})-c}{2}\zeta^2+\frac{1}{4}\frac{g'}{c_0}B \zeta^3\right)\,.
\end{equation}
Defining by 
\begin{equation}\label{adef}
    \zeta_m = \frac{2 c_0(c-c_0(1+\frac{q^2}{2}))}{g'B}\, 
\end{equation}
the non-zero root of the cubic potential, we can rewrite the previous equation \eqref{equa} compactly as
\begin{equation}\label{newt_pot_KdV}
    \zeta_x^2 =\frac{3 B}{\kappa} \zeta^2(\zeta_m-\zeta)
\end{equation}
as in the KdV case~(see \cite{Wh2000}). Similarly to the full system's solitary wave discussion, for solutions to exist the potential's concavity at the origin has to be negative which forces the speed $c$ to be higher than that of the linearized equation,   
\begin{equation}
    c>c_0\left(1+\frac{q^2}{2}\right)\, .
    \label{cc0}
\end{equation}
As in the plane wave solitons of the full system discussed in Section \ref{TPSW-KBKB},  the sign of $\zeta_m$ depends on the sign of $B$: when $B$ is positive, we have a soliton of elevation, while when $B$ is negative the soliton is of depression. 

The explicit line soliton formula  (see, e.g., \cite{Wh2000}) is, as well known,  
\begin{equation}
    \zeta(x,y,t) = \zeta_m\, \mathrm{sech}^2\left(\sqrt{\frac{3 B \zeta_m}{4\kappa}}(x-ct+qy)\right)\, .
    \end{equation}
    We notice that, thanks to \eqref{adef} and the condition \eqref{cc0}, the argument of the hyperbolic function is always positive.
    
A non-trivial restriction on the velocity is however given by the parameters of the problem. Indeed we have to impose {\em a posteriori} that the solitary wave does not contact the physical boundaries, where, incidentally,  the whole formalism breaks down from the very beginning at the level of the parent equations~\cite{CFOP14}. That is, we have to require
\begin{equation}\label{condition_velocity}
\begin{cases}\medskip 
 \zeta_m>  -\dsl{\frac{h_2}{{h}}}\quad \text{for } B<0\\
  \zeta_m<  \dsl{\frac{h_1}{{h}}}\quad  \text{for } B>0\, .
\end{cases}
\end{equation}
These  conditions \ translate into the following constraints on the speed $c$ of  the soliton:
\begin{equation}
\begin{split}
  c_0\Big(1+\frac{q^2}{2}\Big)<  c<c_0\Big(1+\frac{q^2}{2}\Big)\left(1-\frac{h_2 B}{2A{h}(1+\frac{q^2}{2})}\right)\, \quad \text{if} \quad B<0,\\
 c_0\Big(1+\frac{q^2}{2}\Big)<   c<c_0\Big(1+\frac{q^2}{2}\Big)\left(1+\frac{h_1 B}{2A{h}(1+\frac{q^2}{2})}\right)\, \quad \text{if} \quad B>0\, .
    \end{split}
\end{equation}
Therefore, enforcing these ``engineering'' constraints set an upper limit to the soliton speed, and ensures the existence of a range of speeds (and amplitudes)  with which KP-II solitons can describe traveling waves in our two-fluid system. The model inherits the same limitations of weakly nonlinear assumptions briefly addressed at the end of Section \ref{specsol} for the parent model.

\section{Conclusions and perspectives}\label{CP}
In this paper we studied the problem of sharply stratified 2-layer Euler fluids in $3$ dimensions from a Hamiltonian perspective. Our starting point was the Hamiltonian structure for incompressible fluids  with general variable density introduced in \cite{Bow87} (a generalization to $3D$  of that of  Benjamin~\cite{Ben86} for $2D$ settings) which we termed Benjamin-Bowman Hamiltonian structure. We considered the sharply stratified two-layer case as a submanifold  of the general field configuration and followed Marsden-Ratiu~\cite{MR86} to a Hamiltonian reduction that yielded: ({\it i}) the description of the effective free parameters of the theory, i.e., the interface displacement  
$\zeta$ and the components $(\wit{\sigma}, \wit{\tau})$ of the tangential momentum jump, or shear, across the interface   (which can be viewed as a density-weighted vorticity supported on the interface), and 
({\it ii}) a simple reduced structure for these Hamiltonian variables.

We then turned to the problem of expressing  the energy of the system via these free variables. We worked in the long wave asymptotic limit, in which the ratio $\epsilon$ between the typical vertical scales charcterized by $h$ and the horizontal ones by $L$   is small. Contrary to the $2D$ case (see e.g., \cite{Lanbook}), it is known that the long wave asymptotics  is not enough to obtain a local Hamiltonian density. A local theory, which is the one we focussed here, can be obtained in the Weakly Nonlinear approximation, defined by the balance of small parameters $\alpha=\dsl{a}/{h}=O(\eps^2)$, $a$ being a typical wave amplitude.
In this case the resulting system coincides with the one known in the literature  as the Kaup-Broer-Kupershimdt-Boussinesq (KBK-Boussinesq) model for $3D$ water waves. In the above approximations, the relevant Hamiltonian variables are the interfacial displacement and the {\em horizontal} momentum shear vector. An important feature of the two-layer model is that the coefficient 
$$
B=\dsl{\frac{\rho_2h_1^2 -\rho_1 h_2^2}{h\, (\rho_2h_1 + \rho_1h_2)}}
$$ 
of  the non-linear term in the effective Hamiltonian may change sign according to the ``hardware parameters" of the physical system, viz. the layer densities $\rho_j$ and the asymptotic heights $h_j$, $j=1,2$. While this sign-change of nonlinearity already occurs in $2D$, its implications in $3D$ settings are perhaps more remarkable. This was brought forth by considering the simplest solutions to this model. We derived the dispersion relations and then discussed traveling waves, showing how the sign of $B$ affects solitary wave solutions which, just as in the $2D$ case, turned out to be wave of depression  when $B<0$ and of elevation  for $B>0$.  However, we also briefly touched upon the problem of stationary solutions, which naturally arise for nonlinear equations in more that one space dimensions. Again, the parameter $B$ played a crucial role, but curiously its effects are somehow reversed with respect to the traveling wave solutions. Indeed, since for a stationary solution  
$\zeta=-{B}|\bs{\gamma}|^2/2$, 
we see that when $B<0$ the interface is attracted by the upper plate, while for $B>0$ by the lower plate. Further study of this mathematical phenomenon and its possible physical implications will be addressed in a future work.

We finally considered the so-called KP-regime, obtaining by breaking the $SO(2)$ symmetry of the problem and assuming that the horizontal independent  variables $(x,y)$ have different scaling laws along the two horizontal directions, namely $x = Lx^*$ and $y = L'y^*$, with  ${L}/{L'} = O(\eps)$.
Unidirectionalization of the resulting model leads to a Kadomtsev-Petviashvili-II equation for the interface elevation $\zeta$, with $B$ dictating the sign of the nonlinear term. Finally, interpreting the unidirectionalization process as a set of constraints on the dependent variables, we show how the Hamiltonian structures of the KP equation can be obtained from that of the KBK-Boussinesq parent system by means of a Dirac reduction process.

 In summary,   our work has framed the problem of evolution of internal waves in $3D$ settings from a Hamiltonian perspective. Not unexpectedly, the effective equations in the weakly 
 non-linear asymptotics reproduce well known models in the theory of water waves, i.e., the KBK-Boussinesq and the KP-II systems.  These models share the  notable feature, typical of internal waves, of the zero crossing  of the coefficient of the  non-linear terms with respect to variations of the densities and asymptotic layer depths, which signals that higher order nonlinearity would come to play a relevant role. From a theoretical point of view, the challenge is to try to go beyond the WNL approximation. In this respect, the mastering of the non localities inherent to the theory, and their framing within the general Hamiltonian reduction scheme herewith presented, is crucial. From a more applicative standpoint, future work will have to assess how the predictions afforded by the models will be reflected, at least qualitatively,  by actual experiments and those {\em in silico}, with attention paid to (quasi) stationary situations.

\subsection*{Acknowledgments}
 We thank D.\ Noja for useful discussions. 
This project was carried out with support by the National Science Foundation under grants RTG DMS-0943851, CMG ARC-1025523, DMS-1009750, DMS-1517879, DMS-1910824, DMS-2308063, by the Office of Naval Research under grants N00014-18-1-2490, N00014-23-1-22478 and DURIP N00014-12-1-0749.
This project has received fundings by 
the European Union's Horizon 2020 research and innovation programme under the Marie Sk{\l}odowska-Curie grant no 778010 {\em IPaDEGAN}
 and by the PRIN 2022TEB52W-PE1 Project ``The charm of integrability: from nonlinear waves to random matrices."
 We also gratefully acknowledge the financial support for RC's visit to Milano-Bicocca 
 by the GNFM Section of INdAM, under which part of this work was carried out, 
 and the financial support of the project MMNLP (Mathematical Methods in Non Linear Physics) of the
INFN. RC \& MP thank the Department of Mathematics and Applications of the University of Milano-Bicocca, ES thanks the Mathematics Department of UNC, and GF thanks SISSA-Trieste for their hospitality where part of this study was carried out.  We finally thank the anonymous referees for their comments which helped improve our presentation.

\subsubsection*{Data availability statement}
No data was used for the research described in the article.

\appendix
\section{Canonical  versus Benjamin-Bowman's Hamiltonian structure}\label{App A}
The Benjamin-Bowman Hamiltonian structure we used in this paper can be seen as the canonical Zakharov-Kuznetsov (ZK) Hamiltonian structure~\cite{Z85,ZK97,S88} for the Euler equations.
Indeed, the ZK approach makes use, along with the density variable $\rho$, of the so-called Clebsch potentials $(\Phi, \la,\mu)$, that define the Eulerian velocity field by 
\begin{equation}\label{a1}
\mb{v}=\dsl{\frac{\la}{\rho}}\> \bs{\nabla}\mu+\bs{\nabla}\Phi\, .
\end{equation}
In the ZK representation  the pairs $(\rho,\Phi)$ and $(\la,\mu)$ are  canonically conjugated, so that  in the  $4$-tuple  $(\rho, \Phi, \la,\mu)$ (which  we shall later denote as $(\rho,\mb{C})$, i.e.$(C^1,C^2,C^3)=(\Phi, \la,\mu)$) the Poisson tensor is the canonical one
\begin{equation}\label{KZPoi}
{J}_{can}=\begin{pmatrix}
0&1&0&0\\-1&0&0&0\\0&0&0&1\\0&0&-1&0
\end{pmatrix}\, .
\end{equation}
The Benjamin-Bowman variables $(\rho, \bs{\Sigma})$ are related with  the  Clebsch potentials by
\begin{equation}\label{Cl->Sigma}
\bs{\Sigma}=\bs{\nabla}\times (\rho\bm{v})=\bna\rho\times\bna\Phi+\bna\la\times\bna\mu\, .
\end{equation}
\begin{prop}
Let $\mb{F}$  denote the change of coordinates
\begin{equation}\label{chacoo}
(\rho, \Phi, \la,\mu) {\buildrel\mb{F} \over \longrightarrow }\, (\rho, \bs{\Sigma})\, , 
\end{equation}
with $\bs{\Sigma}$ given by \eqref{Cl->Sigma}.

The  canonical Poisson tensor \eqref{KZPoi} and the Benjamin-Bowman  Poisson tensor \eqref{PBow}, i.e.,
\begin{equation}\label{Pbow2}
P_B=
-\left(
\begin{array}{cc}\bigskip
0&\bm{\nabla}\cdot(\rho\, \bm{\nabla}\times\>\> )\\
\bm{\nabla}\times(\rho\, \bm{\nabla}\cdot \>\>)
&\bm{\nabla}\times(\mb{\Sigma}\times \bm{\nabla}\times \> \>)
\end{array}\right)\, \left(
\begin{array}{c}\medskip
\dsl{\frac{\delta \HH}{\delta\rho}}\\
\dsl{\frac{\delta \HH}{\delta \Sigb}}
\end{array}
\right)
\end{equation}
are $\mb{F}$-related.
\end{prop}
{\bf Proof.} The proposition can be proven by  tedious but straightforward computations, whose key points are collected below.

Denoting by  $\mb{F}_*$ the (Fr\'echet) Jacobian of the map  \eqref{chacoo}, and with $\mb{F}^*$ its transpose, we have to prove the equality
\begin{equation}\label{Ptransform}
P_B=\mb{F}_*\cdot {J}_{can}\cdot  \mb{F}^*\, , 
\end{equation}
where $\cdot$ stands for both matrix multiplication and composition of differential operators.

We have  (using Einstein's convention of summation over repeated indices that run from $1$ to $3$) that the Jacobian matrix is
\begin{equation}\label{jacF}
\mb{F}_*=\begin{pmatrix}\medskip
1&0&0&0\\ \medskip
\eps^{1lm}\partial_m\Phi\partial_l&\eps^{1lm}\partial_l\rho\partial_m&\eps^{1lm}\partial_m\mu\partial_l&\eps^{1lm}\partial_l\la\partial_m\\ \medskip
\eps^{2lm}\partial_m\Phi\partial_l&\eps^{2lm}\partial_l\rho\partial_m&\eps^{2lm}\partial_m\mu\partial_l&\eps^{2lm}\partial_l\la\partial_m\\ \medskip
\eps^{3lm}\partial_m\Phi\partial_l&\eps^{3lm}\partial_l\rho\partial_m&\eps^{3lm}\partial_m\mu\partial_l&\eps^{3lm}\partial_l\la\partial_m
\end{pmatrix}
\end{equation}
and, thanks to the antisymmetry of the Levi-Civita $\eps$ symbol, its transpose is 
\begin{equation}\label{jacTF}
\mb{F}^*=\begin{pmatrix}\medskip
1&-\eps^{1lm}\partial_m\Phi\partial_l&-\eps^{2lm}\partial_m\Phi\partial_l&-\eps^{3lm}\partial_m\Phi\partial_l\\ \medskip
0&-\eps^{1lm}\partial_l\rho\partial_m&-\eps^{2lm}\partial_l\rho\partial_m&-\eps^{3lm}\partial_l\rho\partial_m\\ \medskip
0&-\eps^{1lm}\partial_m\mu\partial_l&-\eps^{2lm}\partial_m\mu\partial_l&-\eps^{3lm}\partial_m\mu\partial_l&\\ \medskip
0&-\eps^{1lm}\partial_l\la\partial_m&-\eps^{2lm}\partial_l\la\partial_m&-\eps^{3lm}\partial_l\la\partial_m
\end{pmatrix}\, .
\end{equation}
Carefully considering the operator multiplications  and using the well known identity
\begin{equation}
 \epsilon^{lmk}\epsilon^{kij}=\delta^{l i}\delta^{m j}-\delta^{l j }\delta^{m i}\, , 
\end{equation}
the assertion follows.
\hfill $\square$

\section{The Hamiltonian reduction procedure}\label{App B}

In this section we give some details on the reduction procedure presented in Section \ref{sectHamred}. We consider the submanifold 
$\CS\subset\CM$, 
see \eqref{S/mfld} and \eqref{M-mfld}, and we restrict the Poisson tensor $P_B$, see \eqref{PBB}, on $\CS$.

Our first task is to compute the annihilator $(T\CS)^0$. It is easily seen that it coincides with the kernel of $i^*:T^*\CM\to T^*\CS$ at the points $(\zeta,{\wit{\sigma}},{\wit{\tau}})\in\CS$, where $i:\CS\to \CM$ is the canonical injection. Since the differential $i_*:T\CS\to T\CM$ is given by 
\begin{equation}
\label{tangentS}
\begin{split}
(\dot\zeta,\dot{\wit{\sigma}},\dot{\wit{\tau}})&\mapsto\left((\rho_2-\rho_1)\dot\zeta\delta(z-\zeta), \quad
\dot{\wit{\sigma}} 
\delta(z-\zeta)-{\wit{\sigma}} 
\dot\zeta\delta'(z-\zeta),
\quad \dot{\wit{\tau}} 
\delta(z-\zeta)-{\wit{\tau}} 
\dot\zeta\delta'(z-\zeta),  \right.
\\ & \qquad\left.
\quad (\dot{\wit{\sigma}}\zeta_x+\dot{\wit{\tau}}\zeta_y+\wit{\sigma}{\dot\zeta}_x+\wit{\tau}{\dot\zeta}_y)\delta(z-\zeta)
-\dot\zeta({\wit{\sigma}}\zeta_x+{\wit{\tau}}\zeta_y)\delta'(z-\zeta)\right),
\end{split} 
\end{equation}
one can check that the image of $(\alpha_\rho,\alpha_\sigma,\alpha_\tau,\alpha_\chi)\in T^*\CM$ under $i^*$ is the covector $(\mu_\zeta,\mu_\wit{\sigma},\mu_\wit{\tau})$, where
\begin{eqnarray}
&&\mu_\zeta=(\rho_2-\rho_1)\wit{\alpha_\rho}+\wit{\sigma}\left(\wit{(\alpha_\sigma)_z}-\wit{(\alpha_\chi)_x}\right)
+\wit{\tau}\left(\wit{(\alpha_\tau)_z}-\wit{(\alpha_\chi)_y}\right)
\label{pre-ele-a}
\\
&&\mu_{\wit{\sigma}}=\wit{\alpha_\sigma}+\zeta_x\wit{\alpha_\chi}
\label{pre-ele-b}
\\
&&\mu_{\wit{\tau}}=\wit{\alpha_\tau}+\zeta_y\wit{\alpha_\chi}
\label{pre-ele-c}
\end{eqnarray}
and, as usual, we denoted with a tilde the evaluation $z=\zeta(x,y)$, i.e., the restriction to the interface between the two fluids. Hence we can conclude 
that $(\alpha_\rho,\alpha_\sigma,\alpha_\tau,\alpha_\chi)$ belongs to $(T\CS)^0=\operatorname{Ker} i^*$ if and only if the conditions 
\begin{eqnarray}
&&(\rho_2-\rho_1)\wit{\alpha_\rho}+\wit{\sigma}\left(\wit{(\alpha_\sigma)_z}-\wit{(\alpha_\chi)_x}\right)
+\wit{\tau}\left(\wit{(\alpha_\tau)_z}-\wit{(\alpha_\chi)_y}\right)=0
\label{ele-a}
\\
&&\wit{\alpha_\sigma}+\zeta_x\wit{\alpha_\chi}=0
\label{ele-b}
\\
&&\wit{\alpha_\tau}+\zeta_y\wit{\alpha_\chi}=0
\label{ele-c}
\end{eqnarray}
are fulfilled.

Now we show that ${\cal P}_B(T\CS)^0=\{0\}$, i.e., the kernel of ${\cal P}_B$ contains $(T\CS)^0$ at every point of $\CS$. Let 
$(\dot\rho,\dot\sigma,\dot\tau,\dot\chi)$ be the image of $(\alpha_\rho,\alpha_\sigma,\alpha_\tau,\alpha_\chi)\in(T\CS)^0$ under ${\cal P}_B$. Then, using the expression \eqref{PBB} of ${\cal P}_B$ at the points of $\CS$ yields
\begin{equation}
\label{PB1}
\begin{split}
\dot\rho&=(\rho_1-\rho_2)\left[\zeta_x\left((\alpha_\chi)_y-(\alpha_\tau)_z\right)
+\zeta_y\left((\alpha_\sigma)_z-(\alpha_\chi)_x\right)
-(\alpha_\tau)_x+(\alpha_\sigma)_y\right]\delta(z-\zeta)\\
&=(\rho_1-\rho_2)\left[\left(\wit{\alpha_\sigma}+\zeta_x\wit{\alpha_\chi}\right)_y
-\left(\wit{\alpha_\tau}+\zeta_y\wit{\alpha_\chi}\right)_x\right]\delta(z-\zeta)\\
&=(\rho_1-\rho_2)\left[(\mu_\wit{\sigma})_y-(\mu_\wit{\tau})_x\right]\delta(z-\zeta),
\end{split}
\end{equation}
where we used \eqref{pre-ele-b} and \eqref{pre-ele-c} in the last equality. Then \eqref{ele-b} and \eqref{ele-c} imply that $\dot\rho=0$.
With similar computations, one can show that  
\begin{equation}
\label{PB2}
\begin{split}
\dot\sigma&=-(\mu_\zeta)_y\delta(z-\zeta)+\sigma\left[(\mu_\wit{\sigma})_y-(\mu_\wit{\tau})_x\right]\delta'(z-\zeta)\\
\dot\tau&=(\mu_\zeta)_x\delta(z-\zeta)+\tau\left[(\mu_\wit{\sigma})_y-(\mu_\wit{\tau})_x\right]\delta'(z-\zeta)\\
\dot\chi&=\left[\zeta_y(\mu_\zeta)_x-\zeta_x(\mu_\zeta)_y\right]\delta(z-\zeta)
+\left[(\mu_\wit{\sigma})_y-(\mu_\wit{\tau})_x\right](\wit{\sigma}\zeta_x+\wit{\tau}\zeta_y)\delta'(z-\zeta),
\end{split}
\end{equation}
entailing that $\dot\sigma=\dot\tau=\dot\chi=0$ as a consequence of \eqref{ele-a}, \eqref{ele-b}, and \eqref{ele-c}. Hence we have shown that (\ref{S=N}) holds, so that the reduced Poisson manifold coincides with $\CS$, which is parametrized by the triple $(\zeta, \wit{\sigma},\wit{\tau})$. 

The final step is to find the reduced Poisson tensor on $\CS$. To this aim, we will use the projection $\pi:\CM\to\CS$ given by \eqref{projmap}, 
seeing $\CS$ as a quotient manifold rather than a submanifold of $\CM$. We know that 
$(\alpha_\rho,\alpha_\sigma,\alpha_\tau,\alpha_\chi)\in T^*\CM$ is an extension of $(\mu_\zeta,\mu_\wit{\sigma},\mu_\wit{\tau})$, i.e., 
one of its preimages of under $i^*$, if and only if \eqref{pre-ele-a}, \eqref{pre-ele-b}, and \eqref{pre-ele-c} hold. We also know --- see \eqref{PB1}
and the first two equations in \eqref{PB2} --- that the evaluation of $P_B$ on such a covector in $T^*\CM$ gives
$$
\begin{pmatrix}
\dot\rho\\
\dot\sigma\\
\dot\tau
\end{pmatrix}
=
\begin{pmatrix}
(\rho_1-\rho_2)\left[(\mu_\wit{\sigma})_y-(\mu_\wit{\tau})_x\right]\delta(z-\zeta)\\
-(\mu_\zeta)_y\delta(z-\zeta)+\sigma\left[(\mu_\wit{\sigma})_y-(\mu_\wit{\tau})_x\right]\delta'(z-\zeta)\\
(\mu_\zeta)_x\delta(z-\zeta)+\tau\left[(\mu_\wit{\sigma})_y-(\mu_\wit{\tau})_x\right]\delta'(z-\zeta)
\end{pmatrix}.
$$
Applying the differential of $\pi$, we obtain
\begin{equation}
\label{redux-PB}
\begin{split}
\dot\zeta&=\frac1{\rho_2-\rho_1}\int_{-h_2}^{h_1}\dot\rho\,dz=-(\mu_\wit{\sigma})_y+(\mu_\wit{\tau})_x\\
\dot{\wit{\sigma}}&=\int_{-h_2}^{h_1}\dot\sigma\,dz=-\int_{-h_2}^{h_1}(\mu_\zeta)_y\delta(z-\zeta)\,dz+
\int_{-h_2}^{h_1}\sigma\left[(\mu_\wit{\sigma})_y-(\mu_\wit{\tau})_x\right]\delta'(z-\zeta)\,dz=-(\mu_\zeta)_y,
\end{split}
\end{equation}
where the last integral in the second line vanishes since the interface does not touch the plates, i.e., $-h_2<\zeta<h_1$. 
In the same way, one finds that 
\begin{equation}
\label{redux-PB-bis}
\dot{\wit{\tau}}=(\mu_\zeta)_x,
\end{equation}
so that the reduced Poisson tensor on $\CS$ is given by \eqref{Pred}.
\section{On the Dirichlet-Neumann operator approach}
\label{App C}
In this appendix we compare the approach through the Dirichlet-Neumann operator of \cite{CGK05} and \cite{BSL08} (see, also, e.g.,~ \cite{CG94, CSchS97,CSS92}) to the computations of Subsection \ref{DN_op}. 

Following the seminal approach of \cite{Zak68} one can write the Euler equations for the two-layer case through the potentials $\Phi_i(x,y,z)$, 
\begin{align}
     &\mathbf{\Delta}\Phi_i=0\, \quad \text{in}\quad \Omega_{i}\, , \label{laplacianeq}\\
    &\Phi_{i\,t} + \frac{1}{2}|\mathbf{\nabla }\Phi_i|^2 = -\frac{p}{\rho_i}-g z\, \quad\text{in}\quad  \Omega_i\,, \label{berneq}
    \end{align} 
where $\Omega_i$ is the $i$-th fluid domain,  ${\mb \Delta} = \partial_x^2+\partial_y^2+\partial_z^2\,,\, {\mb \nabla}=(\partial_x,\partial_y,\partial_z)$,  
equipped with  the following set of boundary conditions
\begin{align}
    &\Phi_{i\,z} = 0\, \quad \text{at}\quad z=(-1)^{i-1} h_i\, , \label{bd_lids}\\
    &\zeta_t = \sqrt{1+\zeta_x^2+\zeta_y^2}\>\partial_n\Phi_2\, \quad \text{at} \quad z=\zeta\, ,\label{bd_interface}\\
   &\partial_n\Phi_1= \partial_n\Phi_2\, \quad \text{at} \quad z=\zeta\, ,\label{bd_cont}
 \end{align}
 where $\partial_n$ is the  normal  derivative at the interface associated with the lower fluid, \begin{equation}
     \textbf{n} = \frac{(-\zeta_x,-\zeta_y, 1)}{\sqrt{1+|(\zeta_x, \zeta_y)|^2}}\,, \quad \partial_n  = \textbf{n}\cdot \mathbf{ \nabla }\, .
 \end{equation}
``Trace'' variables are defined as
 $ \wit \Phi_{i} =  \Phi_{i\,|_{z=\zeta}}$. The starting point of the DN approach relies on the observation (see, e.g., \cite{BSL08}) that, 
once given $\wit \Phi_1(x,y)$, we can formally find both the bulk potentials  $\Phi_i$'s. Indeed, knowing $\wit \Phi_1$, we can solve the following Laplace problem with mixed boundary conditions for $\Phi_1$, 
 \begin{equation}\label{laplace_prob}
 \begin{cases}
  \mathbf{   \Delta} \Phi_1 = 0\, \quad \text{in}\quad \Omega_1\\
     \Phi_{1\,z} = 0\, \quad \text{at}\quad z = h_1\\
     \Phi_1 = \wit \Phi_1\, \quad \text{at} \quad z=\zeta\, . 
 \end{cases}
 \end{equation}
 Then, using the boundary condition \eqref{bd_cont}, we can consider the Neumann boundary value problem
 \begin{equation}\label{neum_prob}
     \begin{cases}
       \mathbf{   \Delta }\Phi_2 = 0\, \quad \text{in}\quad \Omega_2\\
     \Phi_{2\,z} = 0\, \quad \text{at}\quad z = -h_2\\
    \partial_n\Phi_2 = \partial_n\Phi_1\, \quad \text{at} \quad z=\zeta\, . 
     \end{cases}
 \end{equation}
 The compatibility condition of this Neumann problem is satisfied thanks to  \eqref{laplace_prob}, so that it has a unique solution up to constants. 
 
 It is known that the problem can be formulated within a canonical formalism, 
  the Hamiltonian  of the problem formally being given by the analogue of the $2$-dimensional counterpart found in~\cite{CGK05}, that is,
 \begin{equation}
\label{HamDNop}
\mathcal{H}[\zeta, \wit{\Phi}_1,\wit{\Phi}_2]=\frac12\int_{\RR^2 }\left(\rho_1\wit{\Phi}_1G_1[\zeta]\wit{\Phi}_1+\rho_2\wit{\Phi}_2 G_2[\zeta]\wit{\Phi}_2+g\, (\rho_2-\rho_1)\zeta^2 \right)dx\,dy\, , 
 \end{equation}
where $G_j[\zeta]$ are the Dirichlet to Neumann operators defined via the relations
\begin{equation}\label{DN/ops}
     G_1[\zeta]\wit \Phi_1=-\partial_n{{\Phi_1}_|}_{z = \zeta}\, ,\quad G_2[\zeta]\wit \Phi_2 = \partial_n{{\Phi_2}_|}_{z = \zeta}\, ,
 \end{equation}
 the minus  sign coming from the orientation of the interface.  
 
To fully describe the problem  in terms of free canonical variables, one needs to find the explicit relation between $\wit \Phi_1$ and $\wit \Phi_2$.  Namely, in the approach of \cite{CGK05}, one has to consider the Laplace problems associated with both fluids \eqref{laplace_prob}  and  \eqref{neum_prob}
 and formulate the problem in terms of the variables $\zeta$ and $ \xi=\rho_2\wit{\Phi}_2-\rho_1\wit{\Phi}_1$ of~\cite{BB97}.
 
 Using the boundary condition \eqref{bd_cont} one arrives at finding the constraint
 \begin{equation}\label{DtoN_rel}
    \nabla \wit \Phi_2 = -\nabla \big( G_2[\zeta]^{-1} \circ G_1[\zeta] (\wit \Phi_1)
    \big) \, .
 \end{equation}
As it customary, one introduces  the small parameters $\alpha, \epsilon$ as in Section~\ref{Asy_ener} in order to find asymptotic expansions of the DN operators, and, in particular, of the constraint \eqref{DtoN_rel}. 

In this asymptotics,  the system \eqref{laplace_prob} becomes
 \begin{equation}\label{laplace_prob_scaled}
    \begin{cases}
   \eps^2  \Delta \Phi_1 +\partial_z^2 \Phi_1 = 0\, \quad \text{in}\quad \Omega_1\\
     \Phi_{1\,z} = 0\, \quad \text{at}\quad z = \frac{h_1}{h}\\
     \Phi_1 = \wit \Phi_1\, \quad \text{at} \quad z=\alpha\zeta\, , 
 \end{cases}
\end{equation}
a similar formulation being given for \eqref{neum_prob}.

As  proven in \cite{CM85, CSchS97}, $G_1[\zeta]$ and $G_2[\zeta]$ admit a series expansions in $\zeta$, so that, taking  our scalings into account, we have 
\begin{equation}
    G_i[\alpha\zeta]=\sum_{j=0}^{\infty} \alpha^jG_i^j[\zeta]\, ,
\end{equation} 
By using the Fourier transform on the first equation of \eqref{laplace_prob_scaled} and on the first equation of  the corresponding system to \eqref{neum_prob},  working in the  WNL  asymptotics we get
\begin{equation}\label{expGop}
    G_i[\zeta] = (-\eps^2\frac{h_i}{h})\left(\Delta+\frac{\eps^2}{3}\frac{h_i^2}{h^2}\Delta^2 +(-1)^i \alpha\frac{h}{h_i}  \nabla \cdot(\zeta \nabla) + O(\alpha^2, \eps^4, \alpha\eps^2)\right)\, , 
\end{equation}
where it might be useful to remark that the operator $\nabla \cdot(\zeta \nabla\, )$ acting on a function $F$, stands for the divergence of the pointwise product of  $\zeta$ with the gradient of $\, F$.

The inversion of  $G_2[\zeta]$, even if in the WNL asymptotics, introduces a non-local term in the picture. Indeed, we can formally write
\begin{equation}\label{G2-1}
    G_2[\zeta]^{-1} = \left(Id-\frac{\eps^2}{3}\frac{h_2^2}{h^2}\Delta - \alpha\frac{h}{h_2}\Delta^{-1}  \nabla \cdot(\zeta \nabla) + O(\alpha^2, \eps^4, \alpha\eps^2)\right)(-\frac{1}{\eps^2}\frac{h}{h_2}\Delta^{-1})\, ,
\end{equation}
whence
\begin{equation}
\begin{split}
        G_2[\zeta]^{-1}\circ G_1[\zeta] &= \frac{h_1}{h_2}\left(Id-\frac{\eps^2}{3}\frac{h_2^2}{h^2}\Delta - \alpha \frac{h}{h_2}\Delta^{-1}  \nabla \cdot(\zeta \nabla) + O(\alpha^2, \eps^4, \alpha\eps^2)\right)\circ\\
       & \circ\left(Id+\frac{\eps^2}{3}\frac{h_1^2}{h^2}\Delta - \alpha \frac{h}{h_1}\Delta^{-1}  \nabla \cdot(\zeta \nabla) + O(\alpha^2, \eps^4, \alpha\eps^2)\right)\\
        &=\frac{h_1}{h_2}+\frac{\eps^2}{3}\frac{h_1}{h_2}(\frac{h_1^2}{h}-\frac{h_2^2}{h})\Delta-\alpha(\frac{h}{h_1}+\frac{h}{h_2}) \Delta^{-1} \nabla \cdot(\zeta \nabla) + O(\alpha^2, \eps^4, \alpha\eps^2)\, .
        \end{split}
\end{equation}
Substituting in the constraint equation \eqref{DtoN_rel} we  get, in the notations of \cite{BSL08}, the final expression of the constraint as 
\begin{equation}\label{nablapsi_rel_WNL}
    \nabla \wit \Phi_2 = -\frac{h_1}{h_2}\nabla\wit \Phi_1-\frac{\eps^2}{3}\frac{h_1}{h_2} \dfrac{h_1^2-h_2^2}{h^2}\Delta \nabla \wit \Phi_1+\alpha h \frac{h_1}{h_2}(\dfrac1{h_1}+\dfrac1{h_2})\, \Pi(\zeta\nabla\wit \Phi_1)+O(\alpha^2, \eps^4,\alpha\eps^2)\, ,
\end{equation}
where  $\Pi$ is the non-local operator 
\begin{equation}
    \Pi:= \Delta^{-1}(\nabla(\nabla \cdot\ ) )\,.
\end{equation}
We can explicitly recover relation \eqref{u2u1} of  Section \ref{DN_op} expressing the interface velocities of the upper fluid in terms of those of the lower fluid as follows. 
First we recall that 
\begin{equation}\label{nocurlwnlorder}
\wit {\mb{u}}_i = \nabla \wit{\Phi}_i+O(\alpha\eps^2)\, , 
\end{equation} 
so that   \eqref{nablapsi_rel_WNL} translates into
\begin{equation}\label{ui_rel_WNL}
     \wit{\mb u}_2 = -\frac{h_1}{h_2}\wit{\mb u}_1-\frac{\eps^2}{3}\frac{h_1}{h_2} \dfrac{h_1^2-h_2^2}{h^2}\Delta \wit{\mb  u}_1+\alpha h \frac{h_1}{h_2}(\dfrac1{h_1}+\dfrac1{h_2})\,\Pi(\zeta\wit {\mb u}_1)+O(\alpha^2, \eps^4,\alpha\eps^2)\, .
\end{equation}
Then we use the identity
\begin{equation}\label{vectrelappa}
\bm{\nabla} (\bm{\nabla }\cdot ( \zeta\wit{\mb u}_1) )= \Delta (\zeta\wit{\mb u}_1 )+ \bm{\nabla} \times \bm{\nabla} \times (\zeta\wit{\mb u}_1)= \Delta (\zeta\wit{\mb u}_1 )+\bm{\nabla}\times(\bm{\nabla}\zeta\times\wit{\mb u}_1)+O(\alpha\eps^2)
\end{equation}
(the second one following from~\eqref{nocurlwnlorder}). Since $\Delta^{-1}$ is a scalar operator and only the first two components of the curl need to be considered for the definition of $\Pi$,   we can write at the WNL order
\begin{equation}\label{Pinostrum}
 \Pi (\zeta  \wit{\mb{u}}_1)=\zeta \wit{\mb{u}}_1 + \bm{\nabla} \times  \Delta ^{-1} (\bm{\nabla} \zeta \times \wit{\mb{u}}_1)\, , 
\end{equation}
yielding the desired equivalence between~\eqref{nablapsi_rel_WNL} and \eqref{u2u1}.

\end{document}